\documentclass[lettersize,journal]{IEEEtran}
 \usepackage{amsmath,amssymb}
 \usepackage{subfigure}
 \usepackage{graphicx,graphics,color,psfrag}
 \usepackage{cite,balance}
 \usepackage{algorithm}
 \usepackage{accents}
 \usepackage{amsthm}
 \usepackage{bm}
 \usepackage{url}
 \usepackage{algorithmic}
 \usepackage[english]{babel}
 \usepackage{multirow}
 \usepackage{enumerate}
 \usepackage{cases}
 \usepackage{stfloats}
 \usepackage{dsfont}
 \usepackage{color,soul}
 \usepackage{amsfonts}
  \usepackage{tcolorbox}
\usepackage[utf8]{inputenc}
 \usepackage{cite,graphicx,amsmath,amssymb}
 \usepackage{subfigure}
 \usepackage{fancyhdr}
 \usepackage{hhline}
 \usepackage{graphicx,graphics}
 \usepackage{array,color}
 \usepackage{amsmath}
 \usepackage{amsthm}

\newtheorem{definition}{Definition}
\newtheorem{proposition}{Proposition}
\newtheorem{remark}{Remark}

\newtheorem{lemma}{Lemma}

\begin{document}
\include{header}

\title{Pruning-Aware Multi-Cluster Co-Inference for Large AI Models in AI-RANs}
\author{Xiaowen Cao, Zhonghao Lyu, Shicheng Chu, Zezhong Zhang, Dingzhu Wen,  Guangxu Zhu, \\ Kaibin Huang, {\it Fellow, IEEE,}  Shuguang Cui, {\it Fellow, IEEE,} and Jie Xu, {\it Fellow, IEEE}
\thanks{X. Cao and S. Chu are with the College of Electronic and Information Engineering, Shenzhen University, Shenzhen  518100, and with Guangdong Provincial Key Laboratory of Future Networks of Intelligence, Shenzhen  518172, China (e-mail: caoxwen@szu.edu.cn, 2400042041@mails.szu.edu.cn). }
\thanks{Z. Lyu is with the Department of Computer Science, The Hang Seng University of Hong Kong, Hong Kong SAR, China (e-mail: lyuz@hsu.edu.hk). }
 \thanks{Z. Zhang, S. Cui, and J. Xu are with the School of Science and Engineering, the Shenzhen Future Network of Intelligence Institute (FNii-Shenzhen), and the Guangdong Provincial Key Laboratory of Future Networks of Intelligence, The Chinese University of Hong Kong (Shenzhen), Shenzhen, Guangdong 518172, China (e-mail: \{zhangzezhong,shuguangcui,xujie\}@cuhk.edu.cn).}
 \thanks{D. Wen is with the School of Information Science and Technology, ShanghaiTech University, Shanghai, China (e-mail: wendzh@shanghaitech.edu.cn).}
 \thanks{G. Zhu is with Shenzhen Research Institute of Big Data, Shenzhen, Guangdong 518172, China (e-mail: gxzhu@sribd.cn).}
 \thanks{K. Huang is with the Department of Electrical and Computer Engineering, The University of Hong Kong, Hong Kong SAR, China (e-mail: huangkb@eee.hku.hk).}
 
\thanks{Part of this work was presented in IEEE International Conference on Communications (ICC) 2026 at Glasgow, Scotland, UK.}
}
\maketitle

\begin{abstract}
The increasing scale and computational demands of large artificial intelligence models (LAIMs) present significant challenges for efficient inference in resource-constrained distributed environments. In this paper, we propose a multi-cluster LAIM co-inference framework, where an edge server equipped with multiple graphics processing units (GPUs) coordinates multiple user clusters to execute inference tasks collaboratively. Within each cluster, devices capture data from diverse perspectives and employ lightweight on-device LAIMs to extract local features. These features are then transmitted to the edge server, where they are aggregated and fused to generate a more accurate inference outcome. To reveal the fundamental trade-off between model pruning and collaborative inference performance, we develop a theoretical framework that characterizes the impact of pruning ratios and device contributions using rate-distortion theory and partial information decomposition. Based on this analysis, we formulate a joint optimization problem that determines the model pruning ratio, the task scheduling strategy, the bandwidth allocation, and the transmission power, with the goal of minimizing the inference distortion while satisfying the constraints of latency, energy consumption, and server capacity. Extensive simulation results demonstrate that the proposed framework significantly outperforms existing benchmark schemes, achieving superior inference accuracy and resource efficiency in multi-cluster edge intelligence networks.
\end{abstract}
\begin{IEEEkeywords}
large AI models, co-inference, edge networks, model pruning, GPU scheduling, resource allocation
\end{IEEEkeywords}
\vspace{-15pt}

\section{Introduction}

As wireless networks evolve toward artificial intelligence radio access networks (AI-RANs) that integrate edge intelligence into network operations \cite{SalmiTNSM_26,Ge2023Enging,Cao26TMC}, large artificial intelligence models (LAIMs), including large language models (LLMs), vision-language models (VLMs), and multimodal foundation models \cite{Zhang24PAMI}, are gradually moving from centralized clouds to network-edge environments. By providing advanced perception, reasoning, and decision-making capabilities, these models enable edge perception networks to process and understand complex environmental information in real time \cite{DongComMAg_LAMBO,lyu2026sense}. As a result, they are becoming an important foundation for emerging applications such as autonomous driving, low-altitude wireless networks, and industrial Internet of Things (IIoT) systems, where timely and intelligent responses are essential \cite{Xu2025Enging}.

Despite their strong performance, the massive size and high computational requirements of LAIMs pose significant challenges for edge deployment \cite{QuCOMST25_LAIM,du2025closing}. Most commercial models are therefore hosted in cloud data centers, which provide sufficient computing resources but suffer from privacy concerns, high communication overhead, and increased latency due to frequent data transmission. To address these limitations, on-device deployment of LAIMs has attracted growing interest. Recent lightweight models, such as Gemini Nano and Llama 2, demonstrate that certain LAIM inference tasks can be executed locally \cite{team2023gemini,touvron2023llama}. Nevertheless, modern foundation models still require substantial memory, computation, and energy resources that exceed the capabilities of most edge devices \cite{Lyu2025JSAC}. Consequently, fully operating large-scale LAIMs on edge devices remains impractical, motivating distributed inference across edge devices and servers.

To this end, \emph{device-edge co-inference} has emerged as a promising paradigm, where a pre-trained model is partitioned between edge devices and servers for joint inference \cite{Li2024TWC_Inference}. By performing lightweight processing locally and offloading subsequent computation to edge/cloud servers, co-inference alleviates on-device resource constraints while enabling access to LAIM services \cite{HeTMC24}. However, practical edge networks often consist of multiple user clusters that simultaneously compete for shared wireless spectrum and graphics processing unit (GPU) resources \cite{Wu25TPD}.  
In such settings, communication, computation, and task scheduling become tightly coupled, and decisions made for one cluster may directly affect the performance of others. Therefore, efficient LAIM deployment requires a unified framework that jointly optimizes task scheduling, resource allocation, and computation management across clusters while satisfying latency and energy constraints.

\subsection{Related Work and Motivation}


Since co-inference relies on feature exchange between edge devices and servers, how to reduce communication overhead is a key challenge. Semantic communications address this issue by transmitting task-relevant features instead of raw data, where deep neural networks (DNNs) are typically used to extract compact semantic representations \cite{10388062,Tilahun24IEEE}. Building on this idea, recent studies have shifted from channel-aware resource allocation to task-oriented designs that account for the importance of transmitted information. 
For instance, \cite{Liu24TCCN} jointly optimized semantic compression and communication resource allocation under task-specific requirements, while \cite{Hu24IoT} considered the joint design of transmit power and bandwidth allocation for semantic text transmission in relay-assisted systems.

Another important research direction focuses on the joint optimization of task offloading and resource allocation. In particular, \cite{Yan22TWC} jointly optimized model placement and model splitting decisions to balance local computation and communication costs, while \cite{Wen23TWC_Inference} investigated the joint allocation of sensing, communication, and computation resources for multi-device inference. Despite these advances, most existing works primarily optimize communication and computation efficiency. In co-inference systems, however, the final inference performance is also influenced by feature quality and the complementary information contributed by different devices. Consequently, the interaction among model compression, resource allocation, and multi-device information fusion remains insufficiently understood.

Nevertheless, most existing studies on co-inference and resource optimization mainly focused on small- or medium-scale AI models, which makes them difficult to directly extend to LAIMs due to their billion-level parameters and high computational cost \cite{Zhang25TWC}. To address the challenges brought by large model size, recent work has explored joint task offloading and resource allocation for LAIM inference. For example, \cite{chen2024LLM} studied how different model splitting points affect end-to-end latency and accuracy under varying wireless channel conditions, showing that the placement of partition layers has a strong impact on system performance. Similarly, \cite{Zhang25IoTJ_LLM} proposed EdgeShard, which partitions a large model into multiple smaller shards that can be deployed across distributed devices to better utilize edge resources, while the authors in \cite{XU26TMC} proposed a large- and small-model  collaborative inference framework that jointly optimizes task offloading and model caching.
These studies highlight the importance of carefully coordinating model partitioning with network conditions in LAIM deployment.


Different from them, model compression is another critical design aspect for LAIM inference. Existing techniques mainly include pruning, quantization, and knowledge distillation, which aim to reduce model size and computational complexity while preserving inference performance. 
Focusing on this, \cite{Lyu2025JSAC} theoretically characterized the relationship between the model pruning and inference distortion, providing a guideline for designing pruning-aware LAIM co-inference systems, which was further explored into model quantization based framework \cite{lyu2026quan}.
The authors in \cite{You26TMC} proposed a joint prompt compression and power control framework for mobile LLM services, where an edge-deployed small language model (SLM) compresses user prompts and a deep reinforcement learning-based policy is adopted to  optimize the compression ratio and transmission power.
However, these designs are developed primarily for single-user co-inference scenarios and therefore do not readily scale to multi-user edge network environments, motivating the present work.

\begin{figure*}[h]
  \centering
  \includegraphics[width=0.9\textwidth]{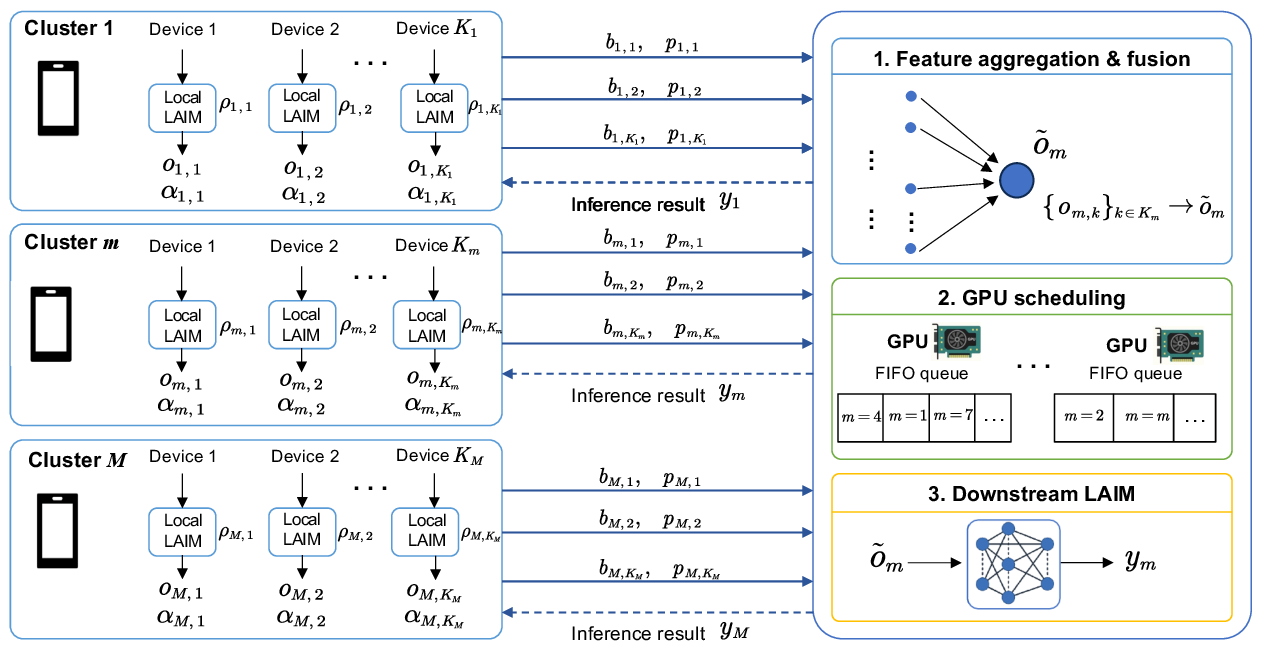}
  \caption{Illustration of multi-cluster LAIM co-inference system.}
  \label{Fig:structure}
\end{figure*}

\subsection{Contributions}

There remain several challenges in enabling efficient LAIM co-inference in edge perception networks. First, multiple user clusters may simultaneously request inference services while sharing limited wireless and edge computing resources, leading to strong resource competition. Second, model pruning can reduce computation and communication costs but may also degrade inference performance, which leaves a fundamental trade-off between resource efficiency and inference quality. Third, different devices often contribute unequally to the final inference outcome due to variations in sensing quality, observation perspectives, and information overlap. Therefore, efficient LAIM co-inference requires the joint consideration of resource allocation, model compression, and device contribution in a unified design framework.
To address these challenges, we propose a multi-cluster LAIM co-inference framework, where an edge server equipped with multiple GPUs coordinates multiple user clusters to perform collaborative inference. The main contributions of this work are summarized as follows.

\begin{itemize}
    \item \textbf{Multi-cluster LAIM co-inference framework}: We consider a novel co-inference architecture in AI-RANs with multiple device clusters. Unlike conventional single-user or single-cluster designs, the proposed framework explicitly captures the competition for shared wireless spectrum and GPU resources among clusters. By incorporating both energy consumption and inference latency into the system design, we characterize the coupling among communication, computation, and task scheduling, thereby providing a practical framework for scalable and resource-efficient LAIM deployment.
    
 \item \textbf{Pruning-performance analysis with heterogeneous device contributions}: To reveal the impact of model compression on collaborative inference, we establish a rate-distortion-based analytical framework that characterizes the fundamental trade-off between model pruning and inference performance. We further incorporate partial information decomposition (PID) to quantify the contributions of different devices. Such analysis provides new insights into how device heterogeneity and multi-view information fusion affect pruning-aware co-inference.

 \item \textbf{Joint optimization for resource-efficient co-inference}: Guided by the proposed pruning-performance analysis, we formulate a joint inference distortion minimization problem that optimizes model pruning ratios, task scheduling, bandwidth allocation, and transmission power under latency and energy constraints. The problem is highly non-convex because model compression decisions directly affect both communication and computation processes, leading to strong interactions among optimization variables. To handle this coupling, we develop an efficient algorithm that enables practical deployment in large-scale edge perception networks.
 
 \item \textbf{Performance evaluation}:
To demonstrate the effectiveness of the proposed framework, we consider two representative edge vision understanding tasks. The results show that the developed pruning-performance analysis efficiently characterizes the inference distortion, while the proposed joint optimization significantly outperforms benchmark schemes. Moreover, by accounting for heterogeneous device contributions, the proposed design effectively improves resource utilization in multi-cluster edge environments.
\end{itemize}

The remainder of this paper is organized as follows. Section II presents the system model. Section III develops the pruning-aware multi-device co-inference quality analysis. Section IV formulates the joint pruning, task scheduling, and communication design. Section V provides the numerical results. Finally, Section VI concludes the paper.

\section{System Model}
We consider LAIM co-inference over a multi-cluster edge system, where a set of user clusters $\mathcal{M}\triangleq\{1,\dots,M\}$ is served by an edge server equipped with $G$ GPUs indexed by $\mathcal{G}\triangleq\{1,\dots,G\}$, as shown in Fig. \ref{Fig:structure}. Let $K_m$ denote the number of edge devices in cluster $m\in\mathcal{M}$.
Cluster $m \in \mathcal{M}$ consists of a set of edge devices $\mathcal{K}_m \triangleq \{1,\dots,K_m\}$ that collaborate to perform a shared inference task. Within each cluster, devices observe the same target or environment from different viewpoints using heterogeneous sensors, enabling multi-view cooperative sensing. This design helps mitigate common limitations of single-device perception, such as occlusions and incomplete local observations, thereby producing more comprehensive and robust feature representations. Meanwhile, clusters are assumed to be task-isolated, with each cluster dedicated to a distinct inference objective or sensing scenario. This assumption eliminates inter-cluster task coupling, prevents cross-task interference, and enables parallel processing of multiple inference tasks on a shared edge server.
Upon receiving a request in cluster $m\in\mathcal{M}$, every device $k\in\mathcal{K}_m$ runs a local LAIM that adapts to limited device computation, to extract embeddings from its raw data and uploads them to the server. 
After gathering all multi-view embeddings from cluster $m$, the server performs intra-cluster feature aggregation and fusion to correct single-view sensing deviations and produce a unified global feature representation, which is then assigned to a single dedicated GPU $g\in\mathcal{G}$ for final inference.
The complete inference workflow of the proposed multi-cluster co-inference framework is described as follows.

\subsection{On-device Inference and Embedding Upload}
First, we discuss pruning-aware on-device inference during forward propagation. Denote the on-device model on edge device $k$ in cluster $m$ as $\mathcal{E}_{m,k}(\cdot;\bm{\theta}_{m,k})$ with parameter vector $\bm{\theta}_{m,k}\in\mathbb{R}^{q_{m,k}}$ with dimension $q_{m,k}$, $\forall k\in\mathcal{K}_m, \forall m\in\mathcal {M}$. To accommodate limited compute/storage, we prune $\bm{\theta}_{m,k}$ into a lightweight model $\widehat{\bm{\theta}}_{m,k}\in\mathbb{R}^{\widehat q_{m,k}}$ with dimension $\widehat q_{m,k}$.
We further define pruning ratio as 
\begin{align}
    \rho_{m,k}\triangleq \widehat q_{m,k}/q_{m,k}\in[0,1], \forall k\in\mathcal{K}_m, \forall m\in\mathcal {M}.
\end{align}
Then the extracted feature of raw input $\bm{x}_{m,k}$ is $\bm{o}_{m,k}=\mathcal{E}_{m,k}(\bm{x}_{m,k};\widehat{\bm{\theta}}_{m,k})$, $\forall k\in\mathcal{K}_m, \forall m\in\mathcal {M}$, where $\mathcal{E}_{m,k}(\cdot)$ denotes the extraction function of the local model at edge device $k\in\mathcal{K}_m, \forall m\in\mathcal{M}$.

Next, we discuss the feature uploading process. 
For ease of illustration, we adopt the frequency division multiple access (FDMA) scheme for wireless data transmission. In this setting, wireless channel conditions stay stable during each data upload process. Each edge device can accurately obtain its own channel state information (CSI), so it can pre-compensate signal phases during signal transmission.
For each device $k\in\mathcal{K}_m$ in cluster $m\in\mathcal{M}$, its uplink transmission of feature data $\bm{o}_{m,k}$ and downlink reception of inference results use an independent and non-overlapping frequency band, with a fixed bandwidth denoted as $b_{m,k}$.
We thus have the bandwidth allocation constraint as
\begin{align}  \sum\limits_{m\in\mathcal{M}}\sum\limits_{k\in\mathcal{K}_m}b_{m,k}\leq B.\label{Sys_Banwidth}
\end{align}
Then the achievable rate of uplink transmission is
\begin{align}\label{eq:uplink-rate}
R_{m,k}(p_{m,k},b_{m,k}) = b_{m,k}\log_2\!\left(1+\frac{g_{m,k}p_{m,k}}{b_{m,k}N_0}\right),
\end{align}
where $g_{m,k}$ is the uplink channel power gain, $N_0$ is the power spectral density (PSD) of the additive white Gaussian noise (AWGN), and $p_{m,k}$ denotes the transmit power with maximum power being $P^{\max}_{m,k}$ given by 
\begin{align}\label{maxpower}
  0\le p_{m,k}\le P^{\max}_{m,k}, \forall k\in\mathcal{K}_m, m\in\mathcal{M}.
\end{align}

\subsection{On-server Aggregation and GPU Scheduling}

Then, we introduce on-server feature aggregation and GPU scheduling. After collecting embeddings $\{\bm{o}_{m,k}\}_{k\in\mathcal{K}_m}$ from all edge devices in cluster $m\in\mathcal{M}$, the server aggregates them as $\widetilde{\bm{o}}_m=\widetilde f_m(\bm{o}_{m,1},\dots,\bm{o}_{m,K_m})$, where ${\tilde f}_{m}(\cdot)$ denotes the aggregation function. 
The aggregated feature $\widetilde{\bm{o}}_m$ is then assigned to one of the available GPUs for on-server inference. Each GPU $g \in \mathcal{G}$ maintains a first-in-first-out (FIFO) processing queue with capacity $\mathcal{Q}_g \triangleq \{1, \dots, Q_g\}$ to manage incoming inference tasks. To model task scheduling, we define a binary variable $x_{m,g,q} \in \{0,1\}$, where $x_{m,g,q} = 1$ indicates that the task of cluster $m$ is assigned to the $q$-th position in the queue of GPU $g$.
Each inference task is scheduled to exactly one position in one GPU queue, and each queue position can accommodate at most one task. This ensures exclusive service allocation and enables a structured representation of GPU workload management for multi-cluster co-inference.
Therefore, we have 
\begin{align}
\sum\limits_{g \in \cal G}\sum\limits_{q \in {\cal Q}_g} x_{m,g,q}=1,\ \forall m\in\mathcal{M}, \label{Sys_Conect1}\\
\sum\limits_{m \in \cal M} x_{m,g,q}\le 1, \forall  q\in \mathcal{Q}_g, g\in\mathcal{G}.\label{Sys_Conect2}
\end{align}

Finally, we discuss the on-server inference process.
Once the task of cluster $m$ is scheduled to a specific GPU, the server executes the remaining inference procedure using the aggregated feature representation $\widetilde{\bm{o}}_m$. Specifically, the GPU processes $\widetilde{\bm{o}}_m$ through a downstream LAIM, and the final inference output is obtained as
\begin{align}
    \bm{y}_m=\widehat f(\widetilde{\bm{o}}_m;\bm{\xi}_m), 
\end{align}
where $\widehat f(\cdot;\bm{\xi}_m)$ denotes the downstream LAIM with parameter $\bm{\xi}_m$. 
The resulting prediction $\bm{y}_m$ is then broadcast from the server to all devices within $\mathcal{K}_m$. This feedback enables each device to obtain a consistent inference result for the shared task, which is particularly important in multi-view cooperative sensing scenarios where devices collaboratively observe the same target or environment. 

Next, we analyze the system delay and energy consumption under the proposed multi-cluster LAIM co-inference framework.

\subsection{System Delay}
The overall system delay consists of three components: on-device inference delay, feature uploading delay,  and on-server inference delay.

\begin{itemize}
    \item  \textbf{On-device inference delay}: We first analyze the arrival time of task $m$, which depends on the on-device inference and transmission delay of $k \in {\cal K}_m$. Denote the original workload of on-device inference (in FLOPs) as $N^{\rm FLOP}_{m,k}$. Assume that the remaining workload is proportional to the pruning ratio $\rho_{m,k}$, and then the remaining workload after model pruning is $\rho_{m,k} N^{\rm FLOP}_{m,k}$. Let $c_{m,k}$ and $f_{m,k}$ denote the number of FLOPs per cycle of the processor and the clock frequency at device $k\in {\cal K}_m$, respectively. Then, the delay for on-device inference is\footnote{Note that the proportional latency and energy model is an analytical approximation rather than a universal hardware law. It captures the basic relationship between the pruning ratio and computational overhead, while providing a tractable abstraction for subsequent system-level optimization. It is particularly reasonable when pruning physically removes computational units and the resulting model structure can be effectively exploited by the underlying execution engine.}
\begin{align}\label{local_inference_delay}
	t_{m,k}^{\rm I}(\rho_{m,k}) = \frac{\rho_{m,k} N^{\rm FLOP}_{m,k}}{f_{m,k} c_{m,k}}.
\end{align}

\item \textbf{Feature uploading delay}: We assume that all features have the same communication overhead of $\beta$ bits. Based on the uplink communication rate in \eqref{eq:uplink-rate}, the delay for embedding $\bm{ o}_{m,k}$ uploading from device $k\in\mathcal{K}_m$ to the server is
\begin{align}\label{communication_delay}
	t^{\rm U}_{m,k}(p_{m,k},b_{m,k}) = \frac{\beta}{R_{m,k}(p_{m,k},b_{m,k})}.
\end{align}
Then the task arrival time of device $k$ is 
\begin{align}
	t^{\rm A}_{m,k}(\rho_{m,k}, p_{m,k},b_{m,k}) = t_{m,k}^{\rm I}(\rho_{m,k})+ t^{\rm U}_{m,k}(p_{m,k},b_{m,k}).
\end{align}
Accordingly, the ready time of task $m$ is
\begin{align}\label{tRm}
t^{\rm R}_m(\{\rho_{m,k},p_{m,k},b_{m,k}\})
\triangleq \max\limits_{k\in\mathcal{K}_m} ~t^{\mathrm{A}}_{m,k}.
\end{align}

\item  \textbf{On-server inference delay}: Then, we consider the on-server inference delay, which includes the queueing and on-server computation delay. Specifically, denote the arrival time at position $q\in \mathcal{Q}_g$ on GPU $g\in\mathcal{G}$ as 
\begin{align}
&A_{g,q}(\{x_{m,g,q},\rho_{m,k},p_{m,k},b_{m,k}\})\!=\notag\\
&~~~~~~~~~~~~~~~~~\sum_{m=1}^{M} x_{m,g,q} t^{\rm R}_m (\{\rho_{m,k},p_{m,k},b_{m,k}\} ).
\end{align}
To enforce FIFO processing, we have
\begin{align}
&A_{g,q}(\{x_{m,g,q},\rho_{m,k}, p_{m,k},b_{m,k}\})\le \nonumber \\
&A_{g,q+1}(\{x_{m,g,q+1},\rho_{m,k},p_{m,k},b_{m,k}\}), \forall  q\in \mathcal{Q}_g, g\in\mathcal{G}.\label{Sys_FIFO}
\end{align} 
Next, let $T_{g,q}$ denote the completion time at position $q\in \mathcal{Q}_g$ on GPU $g\in\mathcal{G}$.
Specifically, denote the on-server workload of task $m$, the number of FLOPs per cycle, and clock frequency at GPU $g$ as ${\hat N}^{\rm FLOP}_m$, $c_g$ and $f_g$, respectively. Let ${\hat t}_m^{\rm I}$ denote the on-server computation delay for task $m$, which is given by 
\begin{align}
    {\hat t}_m^{\rm I}= {\hat N}^{\rm FLOP}_m / {f_g c_g}, \forall m\in\mathcal{M}.
\end{align}
Under the FIFO scheduling discipline, the task at queue position $q$ can start execution only after both the completion of the task at position $(q-1)$ and the arrival of the current task. Therefore, we have
\begin{align} \label{Tgp}
T_{g,q}=& \max \{T_{g,q-1},A_{g,q}(\{x_{m,g,q},\rho_{m,k},p_{m,k},b_{m,k}\})\} \nonumber \\
&+ \sum\nolimits_{m \in \cal M} x_{m,g,q}{\hat t}_m^{\rm I},\forall  q\in \mathcal{Q}_g, g\in\mathcal{G}.
\end{align}

\end{itemize}

\subsection{Energy Consumption}
Similar to the delay analysis, the overall energy consumption can be decomposed into three components: on-device inference energy consumption, feature uploading energy consumption, and on-server inference energy consumption.

\begin{itemize}
    \item  \textbf{On-device inference energy consumption}:
We first consider the energy consumed during on-device inference. Specifically, each edge device processes its local sensing data using the deployed lightweight LAIM with pruning ratio $\rho_{m,k}$ to extract intermediate feature embeddings. 
Let $e^{\mathrm{I}}(\{\rho_{m,k}\})$ denote the on-device inference energy consumed by device $k \in {\cal K}_m, \forall  m\in\mathcal{M}$.
The corresponding energy consumption depends on the required computational workload and the device computing capability, and can be expressed as follows
\begin{align}
e^{\mathrm{I}}(\{\rho_{m,k}\})
=\sum_{m=1}^M\sum_{k\in\mathcal{K}_m}
\eta_{m,k}\,\phi_{m,k}\,f_{m,k}^2\,
\frac{\rho_{m,k} N_{m,k}^{\mathrm{FLOP}}}{c_{m,k}},
\end{align}
where $\eta_{m,k}$ and $\phi_{m,k}$ represent the power usage effectiveness (PUE) and chip-specific power coefficient at device $k \in {\cal K}_m, \forall  m\in\mathcal{M}$, respectively.

    \item  \textbf{Feature uploading energy consumption}:
Next, we consider the energy consumed during feature uploading. After completing on-device inference, each device transmits its extracted feature embeddings to the edge server through wireless links. The corresponding transmission energy denoted by $e^{\mathrm{U}}(\{p_{m,k},b_{m,k}\})$ depends on the transmit power $p_{m,k}$ and allocated communication resources $b_{m,k}$. Therefore, the total feature uploading energy consumption of all devices can be expressed as follows.
\begin{align} \label{powerforuploading}
e^{\mathrm{U}}(\{p_{m,k},b_{m,k}\})
=\sum\limits_{m\in\mathcal{M}}\sum\limits_{k\in\mathcal{K}_m}
\frac{p_{m,k}\,\beta}{R_{m,k}(p_{m,k},b_{m,k})}.
\end{align}

    \item  \textbf{On-server inference energy consumption}:
Then, the overall energy consumption of the server consists of a base power term and a computation power term.
Let $P^{\rm b}$ denote the base power of the server, which accounts for auxiliary components such as lighting and displays. Accordingly, the total server energy consumption is given by
\begin{align}
&\widehat e^{\mathrm{I}}(\{x_{m,g,q},T_{g,q}\})=
\nonumber \\
& P^{\mathrm{b}} \max_{g,q} \{T_{g,q}\}
+ \sum_{m=1}^M\sum_{g=1}^G\sum_{q=1}^{Q_g}\eta_g 
x_{m,g,q}{\phi}_g\,f_g^2\,
\frac{\widehat N_m^{\mathrm{FLOP}}}{c_g},
\end{align}
where  $\eta_g$ and $ \phi_g$ are the PUE and power coefficient of GPU $g$, respectively.

\end{itemize}
Finally, the total energy consumption during the co-inference process is 
\begin{align}
&{\hat E}(\{x_{m,g,q},\rho_{m,k},p_{m,k},b_{m,k},T_{g,q}\})= \nonumber \\
&e^{\mathrm{I}}(\{\rho_{m,k}\})+ e^{\mathrm{U}}(\{p_{m,k},b_{m,k}\})+\widehat e^{\mathrm{I}}(\{x_{m,g,q},T_{g,q}\}).
\end{align}


 \section{Pruning-aware Multi-device Co-inference Quality Analysis}
We next characterize the impact of model pruning on the inference performance of multi-device LAIM co-inference systems. In multi-view perception scenarios, devices often contribute different amounts of useful information due to variations in sensing viewpoints and observation quality. Moreover, the information provided by a device may be unique, redundant with other devices, or only valuable when combined with observations from other views. As a result, the effect of pruning on the overall inference performance depends not only on the compressed model itself but also on the contribution of the corresponding device. To capture this interplay, we first present a rate-distortion relationship for the single-device case and then extend the analysis to the multi-device setting by incorporating heterogeneous device contributions through a PID framework. This provides a theoretical basis for contribution-aware pruning and resource allocation in co-inference systems.

\subsubsection{Rate-distortion Analysis in Single-device Case}
Directly analyzing the output distortion of an LAIM is challenging due to its strong nonlinearity. To overcome this difficulty, we first upper-bound the output distortion by the parameter distortion, as stated in the following lemma.
\begin{lemma}\label{lemma:prune2param}\emph{
Denote $\bm{W}$ and $\widehat{\bm{W}}$ as the original and pruned parameter vectors of an LAIM $f(\cdot)$, respectively. Under the smoothness of the activation function in $f(\cdot)$, for input $\bm{x}$, we have 
\begin{align}\label{eq:prune2param}
\big\|f(\bm{x},\bm{W})-f(\bm{x},\widehat{\bm{W}})\big\|
\le C \big\|\bm{W}-\widehat{\bm{W}}\big\|_{\mathrm{F}},
\end{align}
where $C$ is a constant relative to LAIM parameters \cite{Lyu2025JSAC}.}
\end{lemma}

\begin{remark}\emph{
Lemma~\ref{lemma:prune2param} establishes a direct relationship between parameter distortion and output distortion. Specifically, the output deviation caused by model pruning is upper bounded by the difference between the original and pruned parameter vectors. This result is important because directly analyzing the output distortion of LAIMs is difficult due to their deep and nonlinear structures. In contrast, parameter distortion can be readily quantified from the pruning process.
To this end, parameter distortion provides a practical approximation for the impact of pruning on inference performance.}
\end{remark}
\begin{figure}[t]
    \centering 
	\subfigure[ResNet-152]{\includegraphics[width=0.24\textwidth]{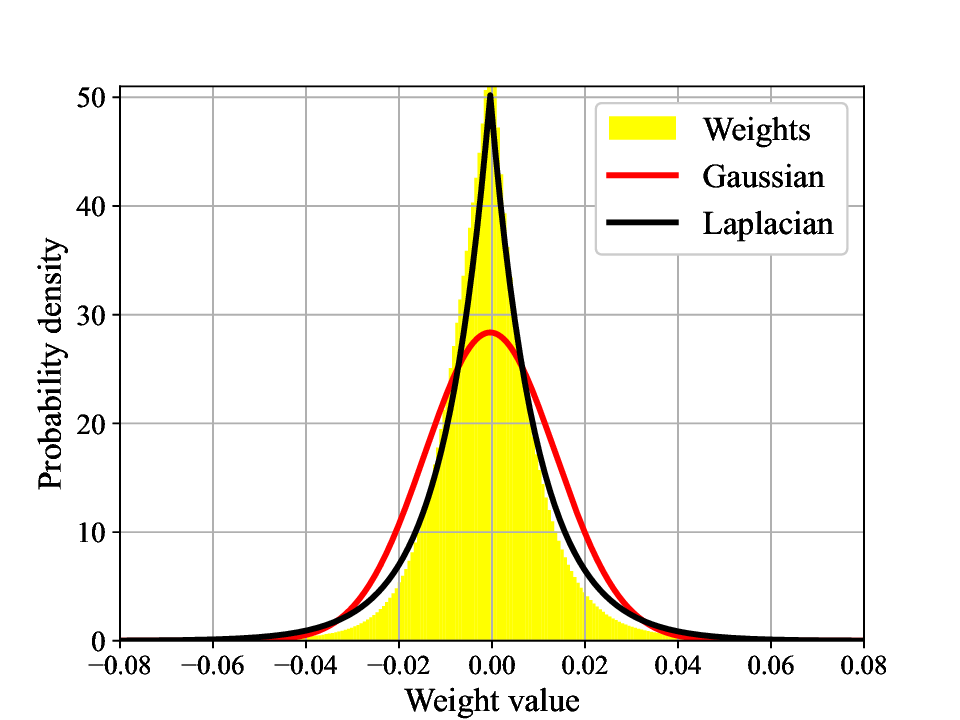}}
	\subfigure[GPT-3]{\includegraphics[width=0.24\textwidth]{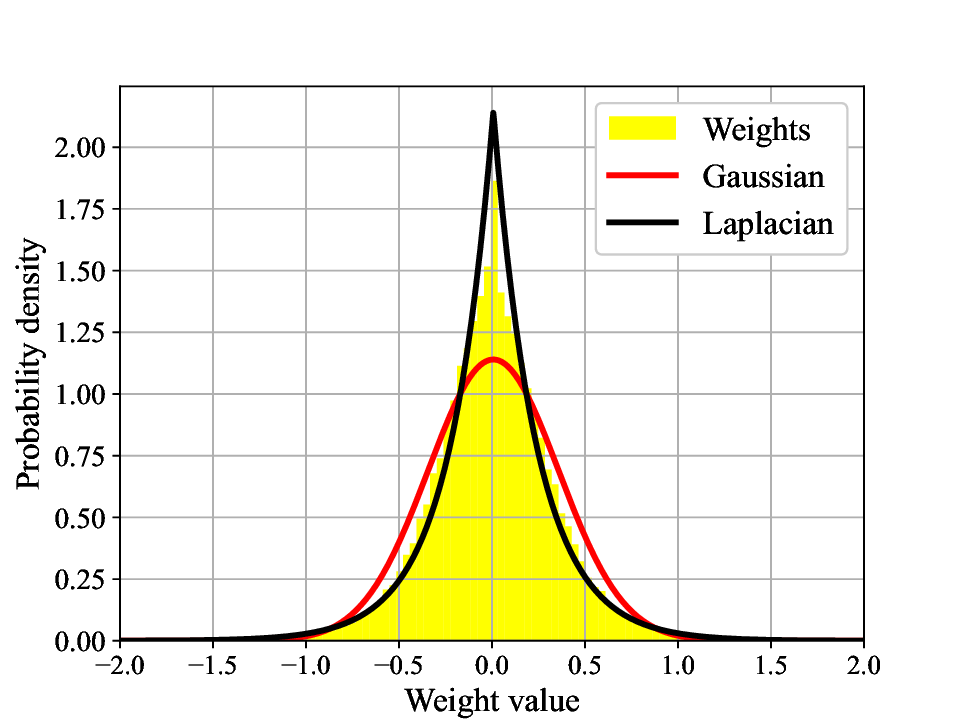}}
    \vspace{5pt}
    \caption{Laplacian distributed weights of ResNet-152 and GPT-3.}
\label{fig:density_weight}
	\vspace{10pt}
\end{figure}

Next, to analyze the relationship between the model pruning ratio and inference performance, we derive the rate-distortion function for LAIM co-inference. 
Define ${\dot R}(D)$ as the rate-distortion function, which equals the infimum of the rate $R$, such that rate-distortion pair $(R, D)$ is achievable \cite{Berger_RateD}.
Bridging this gap via Lemma~\ref{lemma:prune2param}, we instead analyze the rate-distortion function of its upper bound, i.e., the model parameter distortion, as stated in the following lemma.
\begin{lemma}\label{lemma:laplace-converse}\emph{
Assume that the LAIM parameters $\bm{W}\in\mathbb{R}^{q}$ follow a multivariate Laplacian distribution. 
Under the distortion function $\|\bm{W}-\widehat{\bm{W}}\|_{\mathrm{F}} \le D $, from \cite{Lyu2025JSAC}, a lower bound of the rate-distortion function is 
\begin{align}\label{eq:rd-lower}
{\dot R}(D) \ge h(\bm{W})
+ q\log\!\left(\frac{q}{\sqrt{\pi e D}}\right)
+ \log\!\left(\frac{\Gamma(q/2)}{2\,\Gamma(q)}\right),
\end{align}
where $h(\bm{W})$ denotes the differential entropy of $\bm{W}$.}
\end{lemma}

\begin{remark}\emph{
Lemma~\ref{lemma:laplace-converse} characterizes a lower bound on the rate-distortion function for the single-device case. The result establishes a fundamental relationship between the coding rate and parameter distortion, thereby providing a tractable framework for analyzing the impact of model pruning on inference performance. From \eqref{eq:rd-lower}, the required rate decreases as the distortion threshold $D$ increases, indicating the inherent trade-off between model compression and performance preservation. In addition, the bound depends on the differential entropy $h(\bm{W})$, which implies that models with more information-rich parameters are generally more sensitive to compression.}
\end{remark}

However, extending this result to multi-device co-inference is challenging. Unlike the single-device case, the overall inference performance depends not only on the distortion of individual models but also on the interactions among devices. Since different devices may provide different amounts of unique and redundant information, their contributions to the final inference result are inherently heterogeneous and must be explicitly considered in the analysis.

\subsubsection{Rate-distortion Analysis in Multi-device Case}

To capture device heterogeneity, we introduce an importance coefficient $\alpha_{m,k}$ for device $k\in\mathcal{K}_m$ in cluster $m\in\mathcal{M}$.  Incorporating $\alpha_{m,k}$ into the distortion measure enables accounting for the unequal contributions of different devices when enforcing a task-level distortion constraint (i.e., a global distortion budget $D$).
Specifically, we establish the following result.
\begin{proposition}\label{prop_multi-device}
\emph{Let $\bm{W}^{(m,k)} \in \mathbb{R}^{q_{m,k}}$ denote the parameter vector of device $k$ in cluster $m$ and $Q^{\rm tot}_{m}=\sum\limits_{k\in\mathcal{K}_m} q_{m,k}$. 
Under the distortion function 
$D_{m,k} \triangleq \|\bm{W}^{(m,k)} - \widehat{\bm{W}}^{(m,k)}\|_{\mathrm{F}}$ 
and a global importance-weighted constraint $\sum_{k \in {{\cal K}_m}} \alpha_{m,k} D_{m,k} \le D, \forall m$ 
with $\alpha_{m,k} > 0$, the rate-distortion function for each device is lower bounded by
\begin{align}\label{eq:per_device_bound}
{\dot R}_{m,k}(D_{m,k}) \ge & 
h\left(\bm{W}^{(m,k)}\right)
\!+\! q_{m,k} \!\log\!\left(\frac{Q^{\rm tot}_{m}\!\alpha_{m,k}}{\sqrt{\pi e}D}\right)
\!\!\notag\\
&+\log\!\left(\frac{\Gamma(q_{m,k}/2)}{2\Gamma(q_{m,k})}\right),
\end{align}
where $\Gamma(\cdot)$ is the gamma function.
Thus, the importance-weighted distortion-rate function is
\begin{align}\label{eq:weighted_distortion}
\!\!\!\!\!D_{m,k}(R_{m,k}) \!
&\ge\! \!
\frac{\alpha_{m,k} q_{m,k}}{\sqrt{\pi e}}
2^{-\frac{\rho_{m,k} q_{m,k} \gamma - h(\bm{W}^{(m,k)})}{q_{m,k}}}
\!\left(\!
\frac{\Gamma(q_{m,k}/2)}{2\,\Gamma(q_{m,k})}\!
\right)^{\!\!\frac{1}{q_{m,k}}}\!\notag\\
   & \triangleq \!\widehat{D}_{m,k}(\rho_{m,k}),\!\!
\end{align}
where $\gamma$ denotes the number of quantization bits per parameter.}
\end{proposition} 
\begin{proof}
We first consider the optimization problem
\begin{align}
\min_{\{D_{m,k}\ge 0\}} & \ \sum\limits_{k\in\mathcal{K}_m} {\dot R}_{m,k}(D_{m,k})
\notag\\
\text{s.t.}\quad &\sum\limits_{k\in\mathcal{K}_m} \alpha_{m,k} D_{m,k} \le D,
\end{align}
which is convex. Its Lagrangian is
\begin{align} \label{Lagrangian}
\mathcal{L}(\{D_{m,k}\},\lambda)&=\sum\limits_{k\in\mathcal{K}_m}\!\left(h(\bm{ W}^{(m,k)})+q_{m,k}\log\tfrac{q_{m,k}}{\sqrt{\pi e}}\right)\notag\\
&~~~~~+\sum\limits_{k\in\mathcal{K}_m}\!\left(\log\tfrac{\Gamma(q_{m,k}/2)}{2\,\Gamma(q_{m,k})}-q_{m,k}\log D_{m,k}\right)\notag\\
&~~~~~+\lambda\left(\sum\limits_{k\in\mathcal{K}_m} \alpha_{m,k} D_{m,k}-D\right).
\end{align}
By setting the derivative of \eqref{Lagrangian} with respect to (w.r.t.) $D_{m,k}$ as zero, i.e., $\frac{\partial \mathcal L}{\partial D_{m,k}}=-\frac{q_{m,k}}{D_{m,k}}+\lambda \alpha_{m,k}=0$, we have $D_{m,k}^\star={q_{m,k}}/{\lambda\,\alpha_{m,k}}$. Substituting $D_{m,k}^\star$ into $\sum\nolimits_{k} \alpha_{m,k} D_{m,k} = D$ yields $\lambda^{\star}=\frac{Q^{\rm tot}_{m}}{D}$.
Therefore, we have $D_{m,k}^\star=\frac{q_{m,k}}{\alpha_{m,k}}\cdot\frac{D}{Q^{\rm tot}_{m}}$.
\end{proof}

\begin{remark}\emph{
Proposition~\ref{prop_multi-device} extends the single-device rate-distortion analysis to multi-device collaborative inference by explicitly accounting for the heterogeneous contributions of different devices with unequal pruning distortions.
Several observations can be drawn from \eqref{eq:per_device_bound} and \eqref{eq:weighted_distortion}. First, the distortion-rate relationship preserves the fundamental trade-off. A higher coding rate, corresponding to a lower pruning level, leads to a smaller distortion, while aggressive pruning increases the distortion. This indicates that model pruning remains a key factor affecting inference quality in co-inference systems.
Second, the device importance coefficient $\alpha_{m,k}$ directly affects the achievable distortion. For a given coding rate, a device with a larger importance coefficient incurs a larger weighted distortion penalty. Therefore, important devices require more preserved model parameters and should be assigned a lower pruning level to maintain the overall inference quality. In contrast, devices with smaller importance coefficients can tolerate more aggressive pruning with a limited impact on the task-level distortion.
Finally, it also reveals a new coupling between model compression and multi-device information fusion. Since the global distortion constraint depends on the weighted sum of device distortions, the pruning ratio of each device should be jointly designed according to its contribution to the final inference task. This result suggests that uniform pruning across devices is generally suboptimal. Instead, importance-aware pruning can better balance inference accuracy and communication overhead in co-inference systems.}
\end{remark}

\subsubsection{Importance-aware Policy}
Finally, we interpret the importance coefficient $\alpha_{m,k}$ from an information-theoretic perspective. To measure the relative contribution of each device, we use the PID framework, which separates each edge device’s contribution into unique, redundant, and synergistic parts.

\begin{definition}\emph{
For any two edge devices with features $\bm{o}_1$ and $\bm{o}_2$, and a target task $Y$, PID decomposes the mutual information $I(Y;\bm{o}_1,\bm{o}_2)$ into nonnegative components as
\begin{equation}
\begin{aligned}
&I(Y;\bm{o}_1,\bm{o}_2) = \mathcal{R} + \mathcal{U}_1 + \mathcal{U}_2 + \mathcal{S}, \\
&I(Y;\bm{o}_1) = \mathcal{R} + \mathcal{U}_1, \quad 
I(Y;\bm{o}_2) = \mathcal{R} + \mathcal{U}_2,
\end{aligned}
\end{equation}
where $\mathcal{R}$ denotes the redundant information shared by both devices, $\mathcal{U}_k$ represents the unique information contributed by device $k$, and $\mathcal{S}$ corresponds to the synergistic information jointly provided by both devices.}
\end{definition}

While the PID framework provides an intuitive information-theoretic formalism to disentangle redundant, unique, and synergistic task-relevant information across edge devices, exact full PID computation suffers from prohibitive combinatorial complexity with more than two devices (i.e., $|\mathcal{K}_m|>2$). To retain PID’s core interpretive logic of quantifying per-device data contribution while enabling tractable computation for multi-device clusters, we leverage Shapley value (SV)-based information allocation. The Shapley value naturally aligns with PID’s fundamental goal of attributing predictive value to individual sensing nodes, and recovers the exact marginal mutual information defined by PID in the two-device special case, serving as a scalable proxy for multi-node information decomposition\footnote{The proposed contribution coefficients are estimated offline or updated only periodically, using representative training or historical sensing data. Since device viewpoints and sensing configurations typically vary much more slowly than the inference process, the estimated coefficients can be reused across multiple inference tasks. Therefore, their computation is not part of the co-inference pipeline and does not contribute to the latency considered in the proposed delay model.}.
Rooted in cooperative game theory, the SV allocates the total utility by quantifying the marginal contribution of each participant across all possible coalitions. As a result, it offers a principled way to evaluate the relative importance of different participants in a collaborative system \cite{Shapley1952,Ghorbani2019DataSE}.
Specifically, let a subset of devices $\hat{\mathcal{K}}_m \subseteq \mathcal{K}_m$ be treated as a coalition of players with a coalition value $\delta(\hat{\mathcal{K}}_m)$. Given an evaluation dataset with $N$ samples 
$\{(\{o_{\hat{k},i}\}_{\hat{k} \in \hat{\mathcal{K}_m} },\, y_i)\}_{i=1}^N$, 
the coalition value is defined as
\begin{align}
\delta(\hat{\mathcal{K}}_m) = 
-\frac{1}{N}\sum\nolimits_{i=1}^N 
\log p_F\!\left(y_i \mid \{o_{\hat{k},i}\}_{\hat{k} \in \hat{\mathcal{K}}_m} \right),
\end{align}
where $p_F(\cdot)$ denotes the predictive distribution. 
Accordingly, the importance coefficient of device $k\in\mathcal{K}_m$ in cluster $m\in\mathcal{M}$ is obtained via the SV as
\begin{align}\label{eq:shapley}
&\alpha_{m,k} \!\!= \notag\\
&
\sum_{\hat{\mathcal{K}}_m\subseteq\mathcal{K}_m\setminus\{k\}} \!\!\!\!\!\!
\frac{|\hat{\mathcal{K}}_m|!(|\mathcal{K}_m|-|\hat{\mathcal{K}}_m|\!-\!1)!}{|\mathcal{K}_m|!}
\left[\delta(\hat{\mathcal{K}}_m\cup\{k\}) \!-\! \delta(\hat{\mathcal{K}}_m)\!\right],
\end{align}
which can be estimated through Monte Carlo sampling over several random device permutations.

\section{Joint Pruning Ratio, Task Scheduling, and Communication Design}

We aim to minimize the inference distortion by jointly optimizing the on-device pruning ratios $\{\rho_{m,k}\}$, 
transmit powers $\{p_{m,k}\}$, bandwidth allocation $\{b_{m,k}\}$, server scheduling decisions $\{x_{m,g,q}\}$, 
and task completion times $\{T_{g,q}\}$, subject to latency and energy constraints. 
Denoting $\Phi = \{\rho_{m,k}, p_{m,k}, b_{m,k},x_{m,g,q}, T_{g,q}\}$, 
the corresponding optimization problem is formulated as
\begin{subequations}\label{P1}
	\begin{align}
	\text{(P1)}:\mathop {\min }\limits_{\Phi} &  ~\sum\limits_{m \in \cal M} \sum\limits_{k\in\mathcal{K}_m} \widehat D_{m,k}(\rho_{m,k}) \nonumber\\
	\mathrm{s.t.}~&
	\max\limits_{g,q} ~\{T_{g,q}\}  \le T_0 \label{P1-a}\\
	~
	& {\hat E}(\Phi) \le E_0 \label{P1-b}\\
	~
	&	T_{g,q} \!\ge\! T_{g,q-1} \!\!+\!\! \sum\nolimits_{m \in \cal M} x_{m,g,q}{\hat t}_m^{\rm I}, \forall g,q >1 \label{P1-c}\\
	~
	&T_{g,q} \ge A_{g,q}(\{x_{m,g,q},\rho_{m,k},p_{m,k},b_{m,k}\}) \nonumber \\
    &  \qquad + \sum\nolimits_{m \in \cal M} x_{m,g,q}{\hat t}_m^{\rm I},  \forall q\in \mathcal{Q}_g, g\in\mathcal{G}\label{P1-d}\\
	~
	& 0 \le \rho_{m,k} \le 1, \forall k\in\mathcal{K}_m, m\in\mathcal{M} \label{P1-e}\\
	~
	&x_{m,g,q}\in\{0,1\}, \forall m\in\mathcal{M},q\in \mathcal{Q}_g, g\in\mathcal{G} \label{P1-f}\\
	~
	& \eqref{Sys_Banwidth}, \eqref{maxpower}, ~\eqref{Sys_Conect1},~\eqref{Sys_Conect2},~\text{and}~\eqref{Sys_FIFO} \notag,
	\end{align}
\end{subequations}
where \eqref{P1-a} and \eqref{P1-b} represent the total delay and energy budget constraints, respectively;  \eqref{P1-c} and \eqref{P1-d} correspond to the queueing relationship defined in \eqref{Tgp};  \eqref{Sys_Conect1}, \eqref{Sys_Conect2}, and \eqref{P1-f} specify the feasible domain of the scheduling variables $\{x_{m,g,q}\}$. The problem is highly non-convex due to the strong coupling among the optimization variables. 
Furthermore, the inclusion of  $\{x_{m,g,q}\}$ renders problem (P1) a mixed-integer nonlinear programming (MINLP) problem, which is generally non-convex and non-trivial to solve optimally.

Therefore, we adopt an alternating algorithm to efficiently solve the above non-convex problem. In particular, we first optimize the task offloading decisions and completion time $\{x_{m,g,q}, T_{g,q}\}$ under given $\{\rho_{m,k},p_{m,k}, b_{m,k}\}$ by introducing an auxiliary variable to obtain a solvable mixed-integer linear programming (MILP) problem. With the obtained task offloading decisions $\{x_{m,g,q}\}$, we then optimize pruning ratios, transmit power, bandwidth allocation, and completion times $\{\rho_{m,k},p_{m,k},b_{m,k},T_{g,q}\}$, in which the successive convex approximation (SCA) technique is used.

\subsection{Optimization of task offloading decisions and completion time}

To tackle problem (P1), we first optimize the task offloading decisions and completion times $\{x_{m,g,q}, T_{g,q}\}$ under fixed pruning ratios, transmit power, and bandwidth allocation $\{\rho_{m,k},p_{m,k}, b_{m,k}\}$, where the resulting problem becomes a feasibility-checking problem. 
Thus, we focus on minimizing delay plus server energy to create slack feasibility regions for the subsequent optimization, which leads to 
\begin{align}
	\text{(P2)}:\mathop {\min }\limits_{\{x_{m,g,q}, T_{g,q}\}} &~   {\hat E}(\{x_{m,g,q}, T_{g,q}\}) + \max\nolimits_{g,q} \{T_{g,q}\}~\nonumber\\
	\mathrm{s.t.}~~~~& \eqref{Sys_Banwidth} - \eqref{Sys_Conect2}, ~\eqref{Sys_FIFO}, \eqref{P1-c},~ \eqref{P1-d}, {\rm and} ~\eqref{P1-f} \nonumber.
\end{align}
Problem (P2) remains non-convex. 
To address this issue, we introduce an auxiliary variable $\Upsilon$ to relax the term 
$\max\nolimits_{g,q}\{T_{g,q}\}$ by imposing the additional constraints 
$\Upsilon \ge T_{g,q},~\forall g,q$, where Problem (P2) is reformulated as  
\begin{align*}
\min \limits_{\{x_{m,g,q}, T_{g,q},\Upsilon\}} &~   {\hat E}(\{x_{m,g,q}, T_{g,q}\}) + \Upsilon~\nonumber\\
	\mathrm{s.t.}~~~~~&~ \Upsilon \ge T_{g,q},~\forall q\in \mathcal{Q}_g, g\in\mathcal{G} \\
    &\eqref{Sys_Banwidth} - \eqref{Sys_Conect2}, ~\eqref{Sys_FIFO}, \eqref{P1-c},~ \eqref{P1-d}, {\rm and} ~\eqref{P1-f} \nonumber.
\end{align*}
The above problem yields an MILP problem, which can be efficiently solved using standard optimization solvers such as Gurobi \cite{Gurobi}.

\subsection{Optimization of Network Resource Allocation}
Then, we optimize pruning ratios, transmit power, bandwidth allocation, and completion times $\tilde \Phi= \{\rho_{m,k},p_{m,k},b_{m,k},T_{g,q}\}$, under fixed task offloading decision $\{x_{m,g,q}\}$, i.e.,
\begin{align}
	\text{(P3)}:\min\limits_{\tilde{\Phi}} &  ~\sum_{m \in \cal M} \sum_{k\in\mathcal{K}_m} \widehat D_{m,k}(\rho_{m,k}) \nonumber\\
	\mathrm{s.t.}~&
	{\hat E}(\tilde \Phi) \le E_0 \label{P3-a}\\
	~
	& \eqref{Sys_Banwidth}, \eqref{maxpower}, ~ \eqref{Sys_FIFO}, ~\eqref{P1-a}, ~\text{and}~ \eqref{P1-c} \sim \eqref{P1-e}. \nonumber
\end{align}
Problem (P3) is still non-convex due to constraints \eqref{P3-a} and \eqref{Sys_FIFO}.
Specifically, due to the property of perspective function, the term $R_{m,k}(p_{m,k},b_{m,k})$ in \eqref{eq:uplink-rate} is concave w.r.t. both $p_{m,k}$ and $b_{m,k}$ \cite{convex-Bouyd}.
Thus, the non-convexity arises from the term $e^{\mathrm{U}}(\{p_{m,k},b_{m,k}\})$ within ${\hat E}(\tilde{\Phi})$ in \eqref{P3-a}, as well as from the convex term $A_{g,q+1}(\{x_{m,g,q}, \rho_{m,k}, p_{m,k}, b_{m,k}\})$ on the right-hand side of \eqref{Sys_FIFO}.
To tackle this issue, we use SCA technique to approximate a non-convex term as a convex one. 
Specifically, for the $j$-th iteration, we approximate $e^{\mathrm{U}}(p_{m,k},b_{m,k})$ at local points $p^{(j)}_{m,k}$ and $b^{(j)}_{m,k}$ using its first-order Taylor expansion as 
\begin{align} \label{p approximation}
&\frac{p_{m,k}{\beta}}{b_{m,k} \log_2 \left(1+\frac{g_{m,k}p_{m,k}}{b_{m,k}N_0} \right)}  \approx \frac{p_{m,k}^{(j)}{\beta}}{b^{(j)}_{m,k} \log_2 \left(1+\frac{g_{m,k} p_{m,k}^{(j)}}{ b^{(j)}_{m,k} N_0} \right)} +\nonumber \\
&~~v(p^{(j)}_{m,k})(p_{m,k}\!-\!p_{m,k}^{(j)}) +\!u(b^{(j)}_{m,k})(b_{m,k}\!-\!b_{m,k}^{(j)})\nonumber \\
&~~\triangleq \bar \zeta^{(j)}(p_{m,k},b_{m,k}),
\end{align}
where 
\begin{align} 
&v\left(p_{m,k}^{(j)}\right)=\frac{\beta}{b_{m,k}^{(j)} \log_2 \left(1+\frac{g_{m,k} p_{m,k}^{(j)}}{ b_{m,k}^{(j)} N_0} \right)}- \nonumber \\
&\frac{p_{m,k}^{(j)}{\beta}g_{m,k}}{b_{m,k}^{(j)}\left( \log_2 \left( 1+\frac{g_{m,k} p_{m,k}^{(j)}}{ b_{m,k}^{(j)} N_0} \right)\right)^2 \left( b_{m,k}^{(j)} N_0 +g_{m,k} p_{m,k}^{(j)}\right) \ln2},\notag
\end{align}
and 
\begin{align} 
&u(b_{m,k}^{(j)})=\frac{p_{m,k}^{(j)}\beta}{\left(b_{m,k}^{(j)}\right)^2 \log_2 \left(1+\frac{g_{m,k} p_{m,k}^{(j)}}{ b_{m,k}^{(j)} N_0} \right)}\times \nonumber \\
&\left(\frac{p_{m,k}^{(j)}g_{m,k}}{\log_2 \left( 1+\frac{g_{m,k} p_{m,k}^{(j)}}{ b_{m,k}^{(j)}N_0} \right) \left( b_{m,k}^{(j)} N_0 +g_{m,k} p_{m,k}^{(j)}\right) \ln2}-1\right)\notag.
\end{align}

Next, we approximate $t^{\mathrm{R}}_m(\{\rho_{m,k}, p_{m,k}, b_{m,k}\})$ in 
$A_{g,q+1}(\cdot)$, which is convex but non-differentiable as  in~\eqref{tRm}. 
To tackle it, we apply the SCA 
based on Danskin's theorem and utilize its subgradient form. 
For the $j$-th iteration with the local point $\{\rho^{(j)}_{m,k}, p^{(j)}_{m,k},b^{(j)}_{m,k}\}$, let $k_m^{\star}$ denote the active device satisfying
\begin{align}
k_m^{\star} 
\in \arg\max_{k \in \mathcal{K}_m}
\left\{
t^{\mathrm{A}}_{m,k}\left(\rho^{(j)}_{m,k}, p^{(j)}_{m,k},b_{m,k}^{(j)}\right)
\right\}.
\end{align}

Then, $t^{\mathrm{R}}_m(\{\rho_{m,k}, p_{m,k}, b_{m,k}\})$ is linearly approximated as
\begin{align}\label{Alg_SCA}
&t^{\mathrm{R}}_m(\{\rho_{m,k}, p_{m,k}, b_{m,k}\}) \ge 
t^{\mathrm{R}}_m\left(\{\rho^{(j)}_{m,k}, p^{(j)}_{m,k}, b^{(j)}_{m,k}\}\right) \nonumber \\
&\!+\! \underline{v}\left(p_{m,k_m^{\star}}^{(j)}\!\right)\!\left(p_{m,k_m^{\star}} \!-\! p_{m,k_m^{\star}}^{(j)}\right)\!+\!\underline{u}\left(b_{m,k_m^{\star}}^{(j)}\right)\left(b_{m,k_m^{\star}} - b_{m,k_m^{\star}}^{(j)}\right) \nonumber\\
&\!+\!\hat{\underline{v}}\left(\!\rho_{m,k_m^{\star}}^{(j)}\!\right)\!\left(\!\rho_{m,k_m^{\star}} \!- \!\rho_{m,k_m^{\star}}^{(j)}\right)
\triangleq 
\underline{t}_m^{\mathrm{R}^{(j)}}(\{\rho_{m,k}, p_{m,k}, b_{m,k}\}),
\end{align}
where $\hat{\underline{v}}\left(\rho_{m,k_m^{\star}}^{(j)}\right)$, $\underline{v}\left(p_{m,k_m^{\star}}^{(j)}\right)$, and $\underline{u}\left(b_{m,k_m^{\star}}^{(j)}\right)$ denote the partial derivatives of $t^{\mathrm{R}}_m(\{\rho_{m,k}, p_{m,k}, b_{m,k}\})$, respectively.
Specifically, we have
\begin{align}\label{vlow}
&\underline{v}\left(p_{m,k_m^{\star}}^{(j)}\right) =\nonumber \\
&\frac{-\beta g_{m,k_m^{\star}}}
{b_{m,k_m^{\star}}^{(j)}\!\!\left(b_{m,k_m^{\star}}^{(j)}N_0\! +\! g_{m,k_m^{\star}}p_{m,k_m^{\star}}^{(j)}\!\right)\!\!\ln 2}
\!\cdot\!
\frac{1}
{\!\!\Big[\log_2\!\left(\!1\!+\!\tfrac{g_{m,k_m^{\star}}p_{m,k_m^{\star}}^{(j)}}{b_{m,k_m^{\star}}^{(j)}N_0}\!\!\right)\!\Big]^2},\\
&\underline{u}\left(b_{m,k_m^{\star}}^{(j)}\right)=\frac{\beta}{\left(b_{m,k_m^{\star}}^{(j)}\right)^2 \log_2 \left(1+\frac{g_{m,k_m^{\star}} p_{m,k_m^{\star}}^{(j)}}{ b_{m,k_m^{\star}}^{(j)} N_0} \right)}\times \nonumber \\
&\left(\frac{p_{m,k_m^{\star}}^{(j)}g_{m,k_m^{\star}}}{\log_2 \left( 1+\frac{g_{m,k_m^{\star}} }{ b_{m,k_m^{\star}}^{(j)}N_0} \right) \left( b_{m,k_m^{\star}}^{(j)} N_0 +g_{m,k_m^{\star}} p_{m,k_m^{\star}}^{(j)}\right) \ln2}-1\right)\notag,\\
&\hat{\underline{v}}\left(\rho_{m,k_m^{\star}}^{(j)}\right)
= 
{N^{\mathrm{FLOP}}_{m,k_m^{\star}}}/
{f_{m,k_m^{\star}}\,c_{m,k_m^{\star}}}.
\end{align}

Furthermore, we introduce a trust region constraint to tighten the lower bound of $\underline{t}_m^{\mathrm{R}^{(j)}}(\{\rho_{m,k}, p_{m,k}, b_{m,k}\})$ in \eqref{Alg_SCA}  due to the non-convex constraint in \eqref{Sys_FIFO}. Denote $\boldsymbol{p}=[p_{1,1},\cdots,p_{1,K_1},p_{2,1},\cdots,p_{M,K_M}]$ and $\boldsymbol{b}=[b_{1,1},\cdots,b_{1,K_1},b_{2,1},\cdots,b_{M,K_M}]$ as the vectors consisting of all optimization variables $\{p_{m,k}\}$ and $\{ b_{m,k}\}$, respectively. Therefore, the trust region constraint is expressed as
\begin{align}  \left\|\left[\boldsymbol{p},\boldsymbol{b}\right]^T-\left[\boldsymbol{p}^{(j)},\boldsymbol{b}^{(j)}\right]^T\right\|_2\leq \Gamma_{\text{tru}}, \label{Alg_SCA_Tru}
\end{align}
with $\boldsymbol{p}^{(j)}=\left[p_{1,1}^{(j)},\cdots,p_{1,K_1}^{(j)},p_{2,1}^{(j)},\cdots,p_{M,K_M}^{(j)}\right]$, $\boldsymbol{b}^{(j)}=\left[b_{1,1}^{(j)},\cdots,b_{1,K_1}^{(j)},b_{2,1}^{(j)},\cdots,b_{M,K_M}^{(j)}\right]$
and $\Gamma_{\text{tru}}$ denoting the trust region radius and $T$ representing the transpose operation.
Accordingly, the convex version  of problem (P3) at the $j$-th iteration is formulated as
\begin{subequations}\label{P4}
\begin{align}
\text{(P4.$j$):}~ 
\min_{\tilde{\Phi}}~~
& \sum_{m \in \mathcal{M}} \sum_{k \in \mathcal{K}_m} 
\widehat{D}_{m,k}(\rho_{m,k})  \nonumber\\
\mathrm{s.t.}~~
& e^{\mathrm{I}}(\{\rho_{m,k}\})
\!+\! \sum_{m \in \cal M} \sum_{k \in {\cal K}_m} \bar{\zeta}^{(j)}(p_{m,k},b_{m,k}) \nonumber \\
& \qquad + \hat{e}^{\mathrm{I}} (\{T_{g,q}\}) 
\!\!\le\! E_0 \notag\\
& A_{g,q}(\{x_{m,g,q}, \rho_{m,k}, p_{m,k}, b_{m,k}\})
\le \nonumber \\
& \sum_{m=1}^{M} \!\!x_{m,g,q+1}\,
\underline{t}_m^{\mathrm{R}^{(j)}}\!(\{\rho_{m,k}, p_{m,k},b_{m,k}\}), \forall g,q \notag\\
& \eqref{Sys_Banwidth}, \eqref{maxpower}, ~\eqref{P1-a},~ \eqref{P1-c} \sim \eqref{P1-e},~\text{and}, \eqref{Alg_SCA_Tru}.\nonumber
\end{align}
\end{subequations}
The above problem is convex, which can be directly solved via CVX \cite{cvx}.


\section{Numerical Results}\vspace{-0.1cm}
In this section, we evaluate the proposed pruning-aware multi-cluster LAIM co-inference framework through two representative tasks: image classification and video-to-text captioning.
All experiments are implemented on NVIDIA RTX~3090 GPUs. Typical parameters are summarized in Table \ref{table_parameters}. The PSD is $N_0=2\times10^{-19}$~W/Hz. 
Wireless channels are modeled as independent and identically distributed (i.i.d.) Rayleigh fading with an average path loss of $10^{-5}$. 
Each cluster contains three partial views with horizontal ranges $[0,0.2]$, $[0.2,0.6]$, and $[0.5,1.0]$, respectively.

\begin{table*}[htbp]
\centering
\caption{SIMULATION PARAMETERS}
\label{table_parameters}
  \renewcommand{\arraystretch}{1.2} 
\begin{tabular}{|c|l|c||c|l|c|}
\hline
\textbf{Parameter} & \textbf{Description} & \textbf{Value} &\textbf{Parameter} & \textbf{Description} & \textbf{Value} \\ \hline
$M$ & Number of clusters & 8 &
$K_m$ & Number of devices in each cluster & 3 \\ \hline
$G$ & Number of GPUs & 2 &
$Q_g$ & Queue length  & 8  \\ \hline
$f_{m,k}$ & On-device clock frequency  & 0.3~GHz  &
$c_{m,k}$ &  Number of FLOPs per cycle &32 \\ \hline
$f_g$ & GPU clock frequency & 1.0,0.85 ~GHz  &
$c_g$ &  Number of FLOPs per cycle of GPU & 128,140 \\ \hline
$P^{\rm max}_{m,k}$ &  Maximum transmit power  & 1~W &
$B$ &  Communication bandwidth  & 5~MHz \\ \hline
$\eta_g$ &   PUE & 2 &$ \phi_g$ & Power coefficient & $\{1,1.2\}\times10^{-28}~\rm{W/(cycle/s)^3}$ \\ \hline
\end{tabular}
\end{table*}
\vspace{10pt}

For the image-classification task, we adopt a vision transformer (ViT)-Tiny model based on the ViT architecture and split it at the third transformer layer for on-device and on-server deployment\footnote{It is noted that the experimental models are relatively lightweight compared with state-of-the-art foundation models. This choice enables controlled evaluation of the proposed analytical framework and optimization algorithm, while the developed framework remains applicable to larger LAIMs. In practical deployment, additional bottlenecks such as memory capacity, slow data retrieval, and KV-cache management also affect inference performance \cite{QuCOMST25_LAIM}. These factors are complementary to the communication and computation optimization considered in this work. Extending the proposed framework to explicitly incorporate memory constraints and memory access overhead is an important direction for future work.}.
Each cluster processes an independent CIFAR-10 inference request, where each image is partitioned into three local views following the method in \cite{Shao-TWC}, with cropping ratios of $0.2$, $0.4$, and $0.6$, respectively, and partially overlapping between the last two views.  The plotted energy and inference latency thresholds are $E_0\in[1.15,1.35]$~J and $T_0\in[1.32,1.40]$~s, respectively.

For the video-to-text task, we use MSR-VTT \cite{MSRVTT} and sample $8$ frames with resolution $224\times224$ from each video.
Each device uses a TinyViT encoder \cite{TinyViT}, and the server uses Qwen3-4B-Base \cite{Qwen3} with a learned projector and a low-rank adaptation (LoRA) adapter to generate captions from the fused view features.
The plotted energy and inference latency thresholds are $E_0\in[3.72,3.92]$~J and $T_0\in[10.2,11.8]$~s, respectively.
For captioning evaluation, we report bilingual evaluation understudy-4 (BLEU-4) \cite{BLEU} and  consensus-based image description evaluation (CIDEr) \cite{CIDEr}. BLEU-4 measures the modified 4-gram precision between a generated caption and its reference captions, with a brevity penalty to discourage overly short outputs. In our implementation, BLEU-4 is computed by averaging smoothed sentence-level BLEU scores over the evaluated samples.
CIDEr evaluates caption quality from a reference-consensus perspective by comparing TF-IDF-weighted $n$-gram representations of generated and reference captions, thereby assigning lower weights to common words and higher weights to more informative terms. The reported CIDEr follows the project evaluation script as a TF-IDF-weighted $n$-gram consensus score over the available references.

To further expose the role of GPU scheduling, we also conduct a server-workload-heterogeneous latency stress experiment for the task-level timeline visualization. This stress experiment uses the same multi-cluster system but increases the communication bandwidth to $120$~MHz for CIFAR-10 and $160$~MHz for video-to-text captioning, so that the timeline is mainly shaped by GPU queueing and on-server computation rather than by feature-uploading delays.
The corresponding energy and inference latency thresholds are $(E_0,T_0)=(1.2~{\rm J},3.0~{\rm s})$ for CIFAR-10 and $(E_0,T_0)=(8.5~{\rm J},22.0~{\rm s})$ for video-to-text captioning.
The on-server workloads are set to $\{8,120,8,8,8,140,130,150\}$~GFLOPs for CIFAR-10 and $\{220,900,260,240,280,940,920,960\}$~GFLOPs for video-to-text captioning.
This stress setting is used only for the latency-decomposition figures and is separated from the quality-threshold sweeps in Figs.~\ref{fig:vit_accuracy}--\ref{fig:video_cider}.

We consider the following benchmark schemes for comparison.
\begin{itemize}
\item  1) \textbf{Fixed power:} All devices transmit at the maximum power;
\item  2) \textbf{Fixed bandwidth:} The total bandwidth is equally allocated among all devices;
\item  3) \textbf{Round-robin scheduling:} Each cluster is sequentially assigned to GPUs in cyclic order without considering queue states;
\item 4) \textbf{Random scheduling:} Tasks are randomly assigned to GPUs and queue positions;
\item 5) \textbf{Equal importance:} Device importance coefficients are set uniformly as $\alpha_{m,k}=1/K_m, \forall m, k$.
\end{itemize}



\begin{figure*}[t]
    \centering
    \subfigure[Energy threshold]{%
        \includegraphics[width=.48\textwidth]{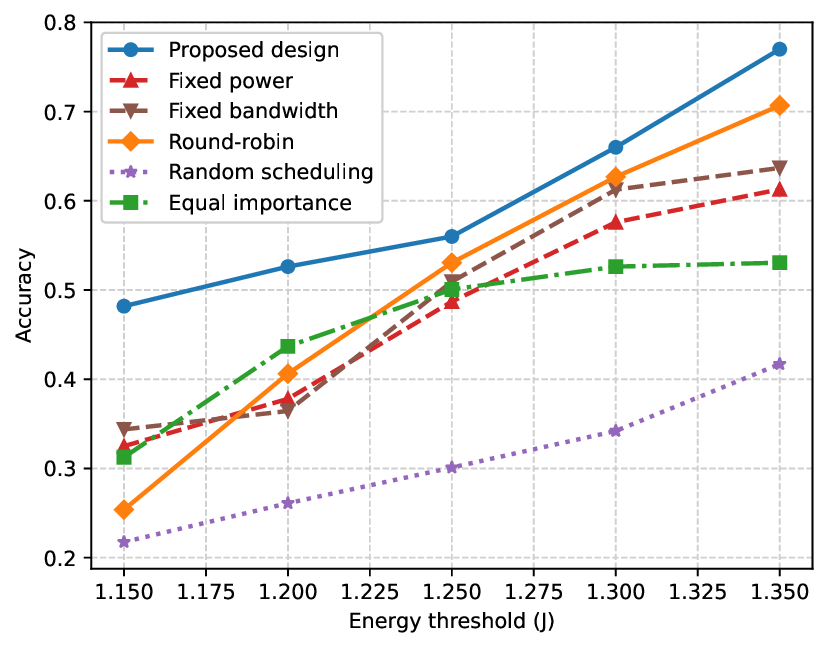}
        \label{fig:vit_accuracy_energy}}
    \hfill
    \subfigure[Inference latency threshold]{%
        \includegraphics[width=.48\textwidth]{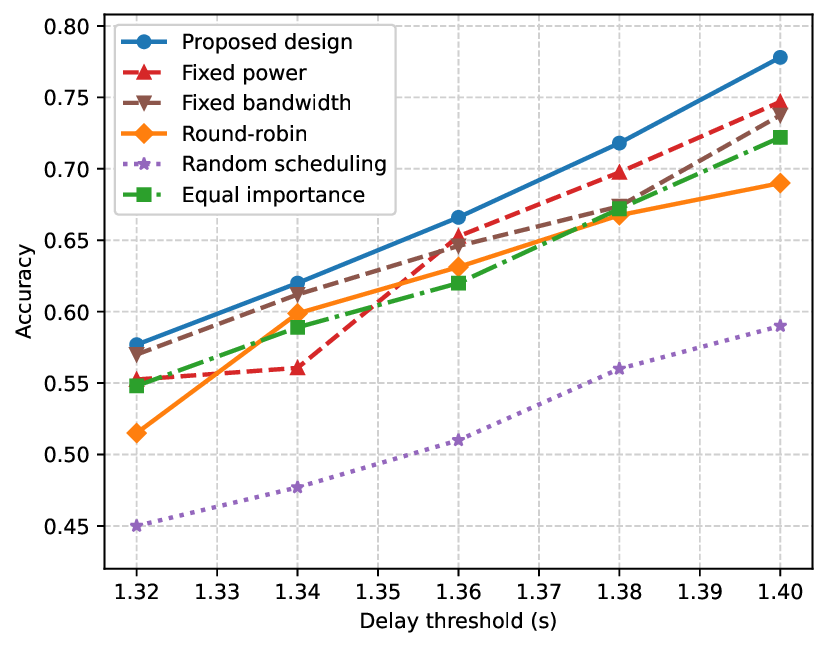}
        \label{fig:vit_accuracy_delay}}
    \caption{Classification accuracy on CIFAR-10 under different energy and inference latency thresholds.}
    \label{fig:vit_accuracy}
    \vspace{10pt}
\end{figure*}

\begin{figure*}[t]
    \centering
    \subfigure[Energy threshold]{%
        \includegraphics[width=.48\textwidth]{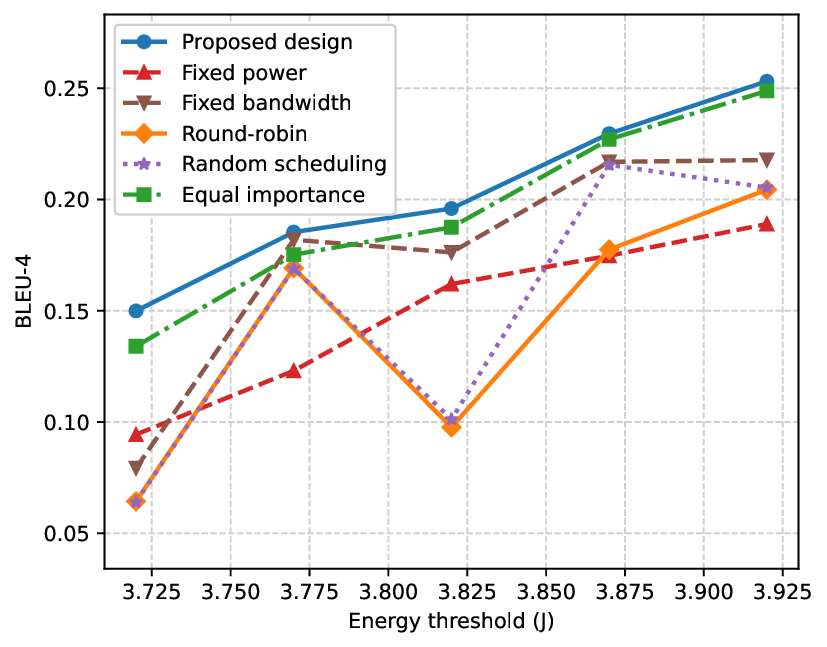}
        \label{fig:video_bleu_energy}}
    \hfill
    \subfigure[Inference latency threshold]{%
        \includegraphics[width=.48\textwidth]{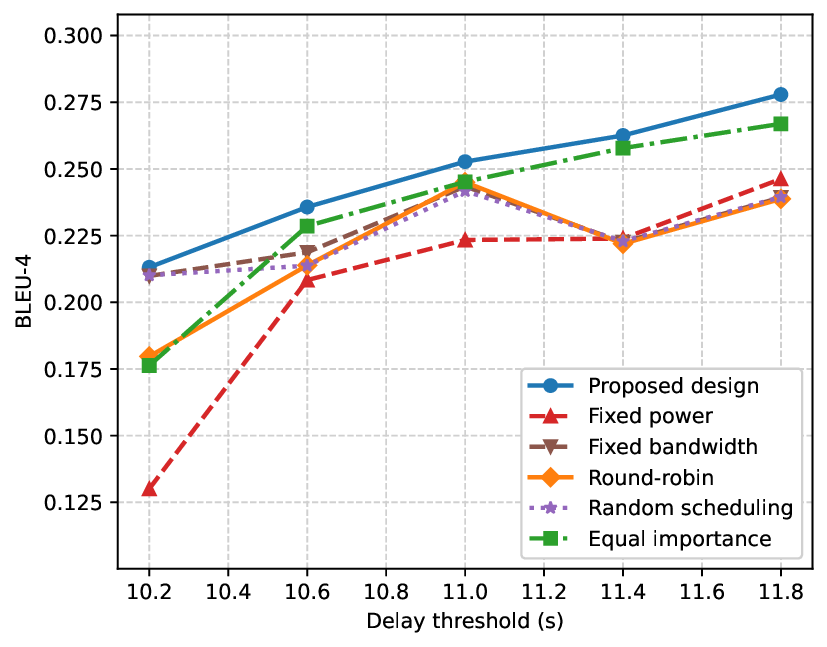}
        \label{fig:video_bleu_delay}}
    \caption{BLEU-4 scores for the video-to-text captioning task under different energy and inference latency thresholds.}
    \label{fig:video_bleu}
    \vspace{10pt}
\end{figure*}

\begin{figure*}[t]
    \centering
    \subfigure[Energy threshold]{%
        \includegraphics[width=.48\textwidth]{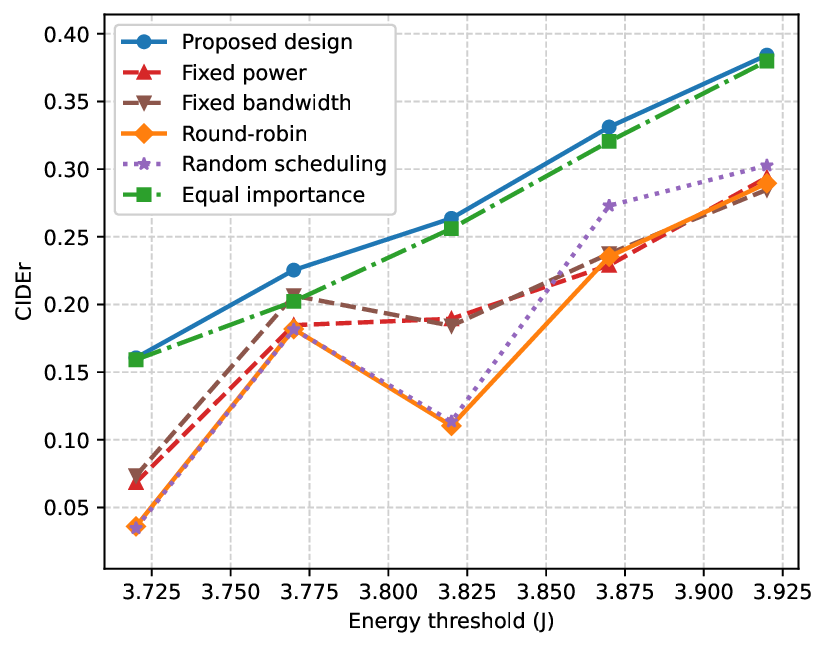}
        \label{fig:video_cider_energy}}
    \hfill
    \subfigure[Inference latency threshold]{%
        \includegraphics[width=.48\textwidth]{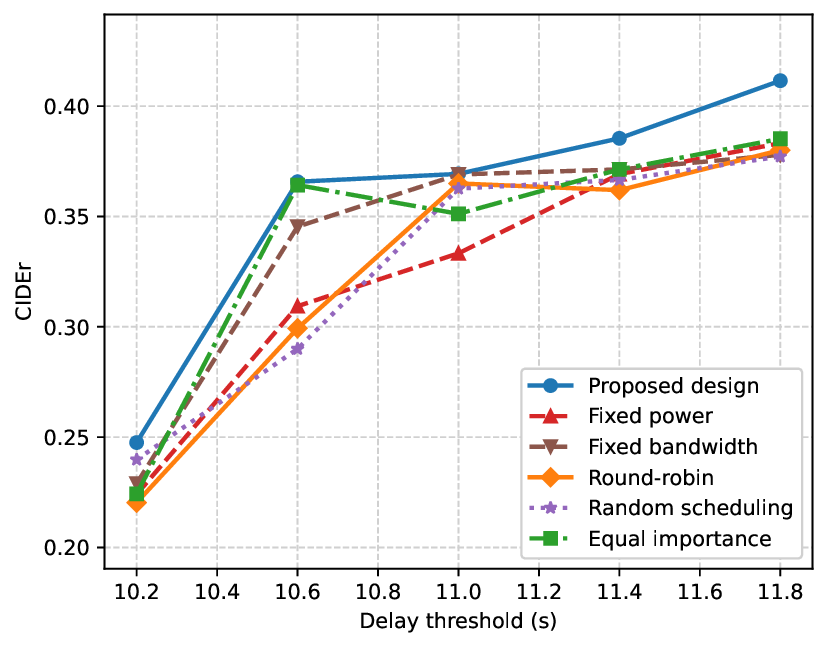}
        \label{fig:video_cider_delay}}
    \caption{CIDEr scores for the video-to-text captioning task under different energy and inference latency thresholds.}
    \label{fig:video_cider}
    \vspace{10pt}
\end{figure*}

\begin{figure*}[t]
    \centering
    \subfigure[Proposed design]{%
        \includegraphics[width=.32\textwidth]{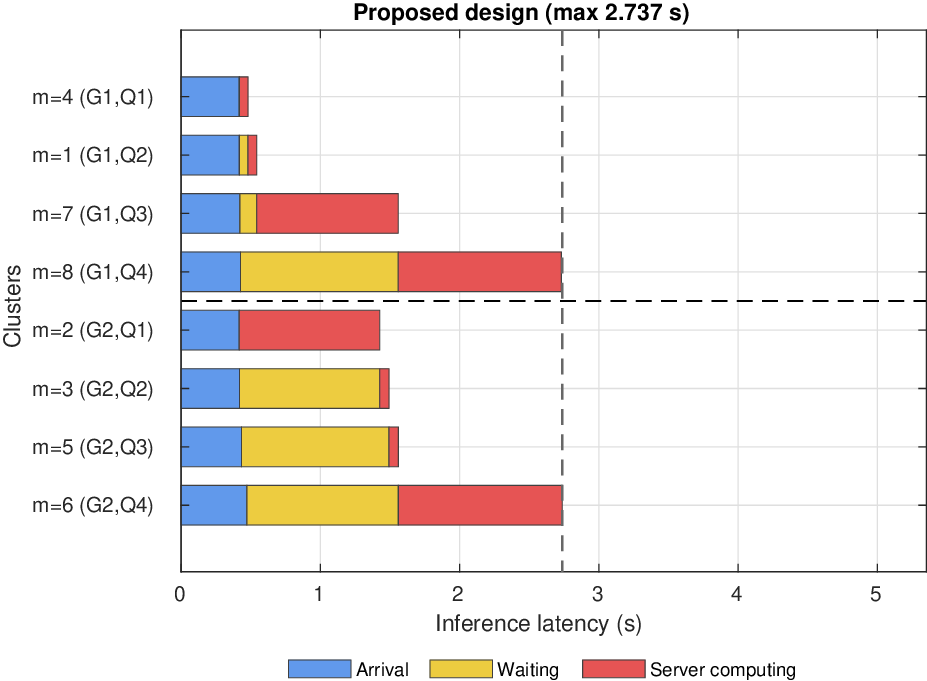}
        \label{fig:vit_delay_timeline_proposed}}
    \hfill
    \subfigure[Round-robin scheduling]{%
        \includegraphics[width=.32\textwidth]{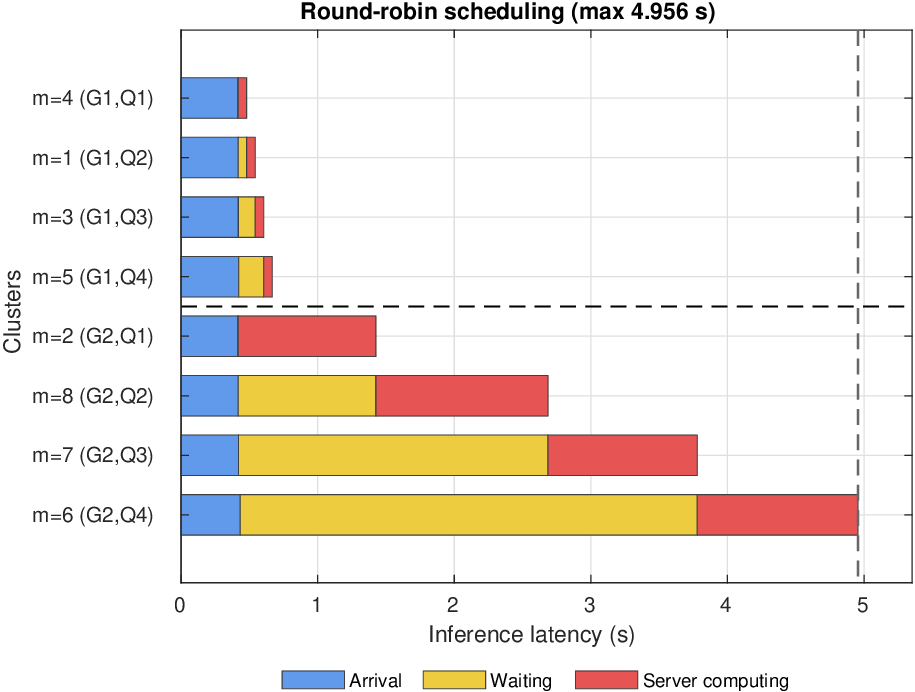}
        \label{fig:vit_delay_timeline_round_robin}}
    \hfill
    \subfigure[Random scheduling]{%
        \includegraphics[width=.32\textwidth]{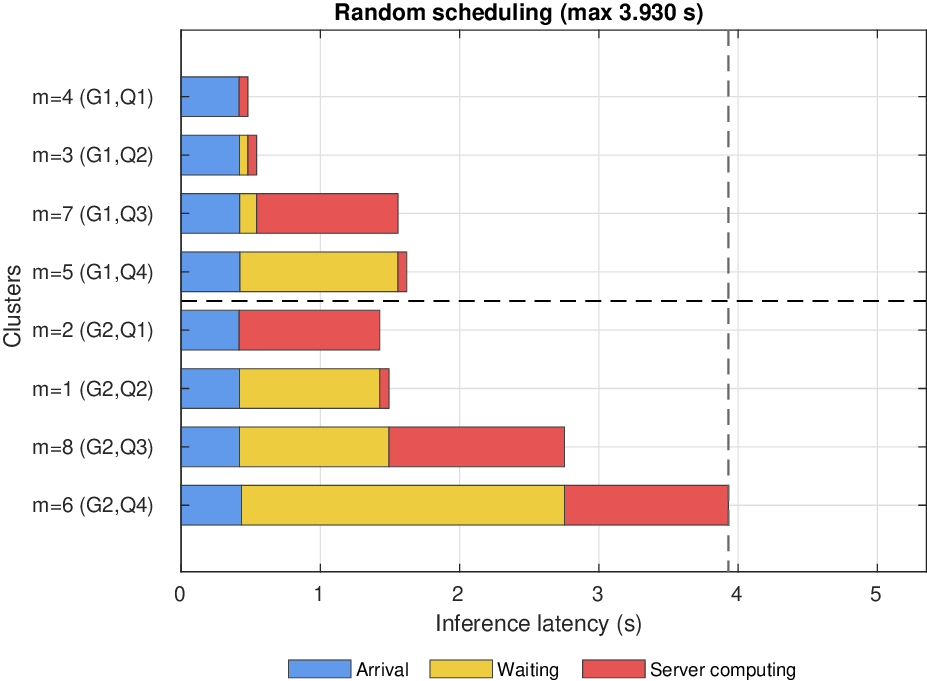}
        \label{fig:vit_delay_timeline_random}}
    \caption{Task-level inference latency decomposition for the CIFAR-10 image-classification task under the server-workload-heterogeneous latency stress setting. Each horizontal bar represents one cluster and is decomposed into arrival, waiting, and server computing intervals. The three panels compare the proposed design, round-robin scheduling, and random scheduling, respectively. The dashed vertical line denotes the maximum cluster completion time in each panel.}
    \label{fig:vit_delay_timeline}
    \vspace{15pt}
\end{figure*}

\begin{figure*}[t]
    \centering
    \subfigure[Proposed design]{%
        \includegraphics[width=.32\textwidth]{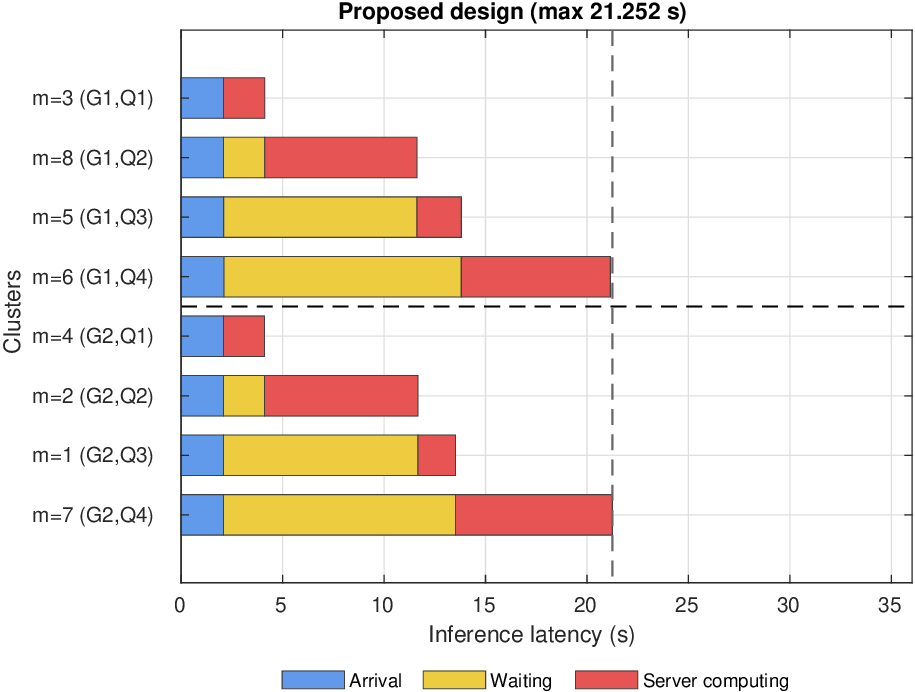}
        \label{fig:video_delay_timeline_proposed}}
    \hfill
    \subfigure[Round-robin scheduling]{%
        \includegraphics[width=.32\textwidth]{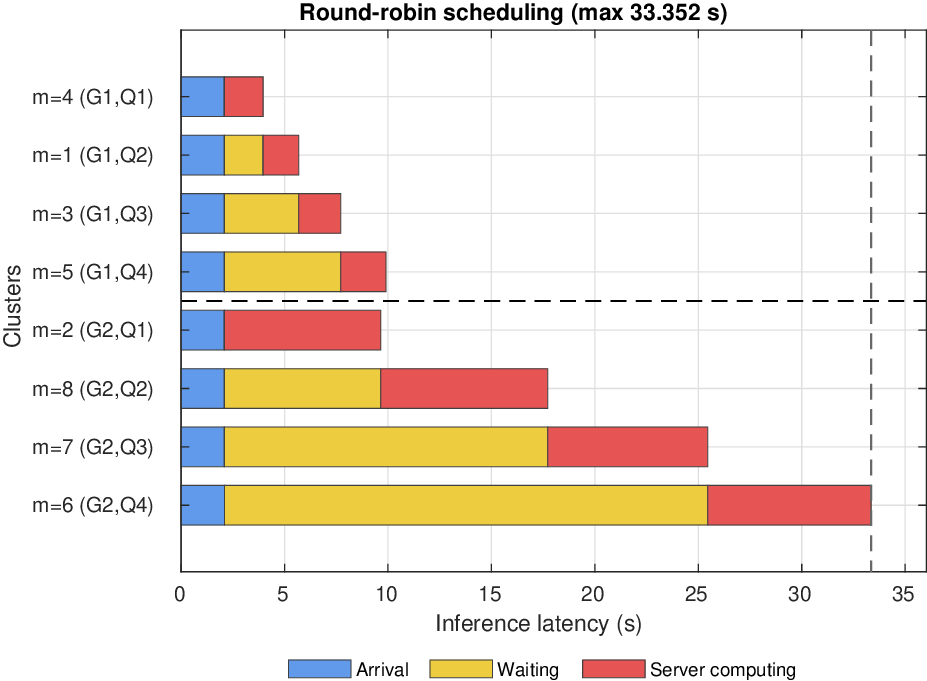}
        \label{fig:video_delay_timeline_round_robin}}
    \hfill
    \subfigure[Random scheduling]{%
        \includegraphics[width=.32\textwidth]{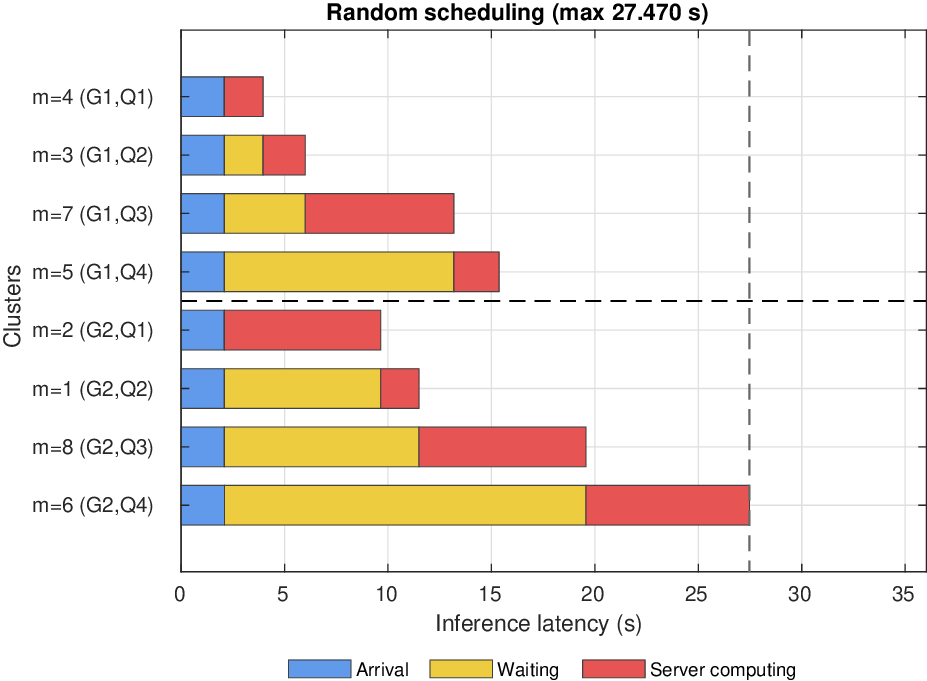}
        \label{fig:video_delay_timeline_random}}
    \caption{Task-level inference latency decomposition for the video-to-text captioning task under the server-workload-heterogeneous latency stress setting. Each horizontal bar represents one cluster and is decomposed into arrival, waiting, and server computing intervals. The three panels compare the proposed design, round-robin scheduling, and random scheduling, respectively. The dashed vertical line denotes the maximum cluster completion time in each panel.}
    \label{fig:video_delay_timeline}
    \vspace{10pt}
\end{figure*}

\begin{figure*}[t]
    \centering
    \subfigure[Cluster $m=3$]{%
        \includegraphics[width=.235\textwidth]{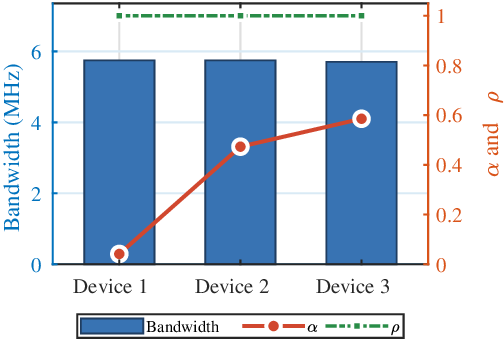}        \label{fig:m3_bandwidth_pruning_insight}}
    \hfill
    \subfigure[Cluster $m=4$]{%
        \includegraphics[width=.235\textwidth]{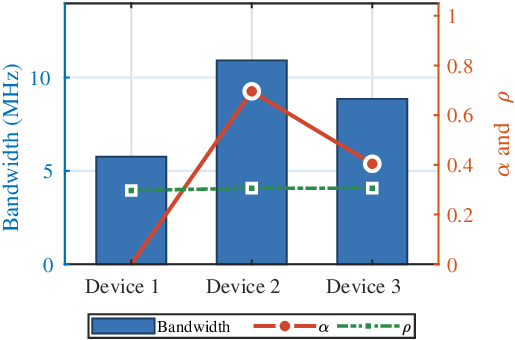}        \label{fig:m4_bandwidth_pruning_insight}}
    \hfill
    \subfigure[Cluster $m=5$]{%
        \includegraphics[width=.235\textwidth]{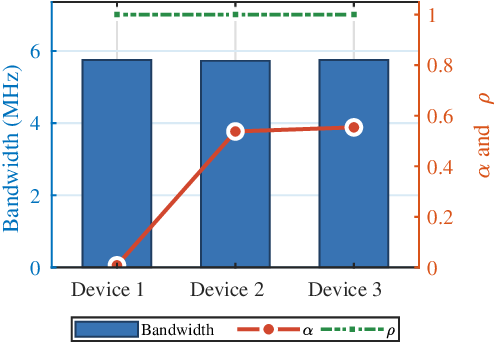}
      \label{fig:m5_bandwidth_pruning_insight}}
    \hfill
    \subfigure[Cluster $m=8$]{%
        \includegraphics[width=.235\textwidth]{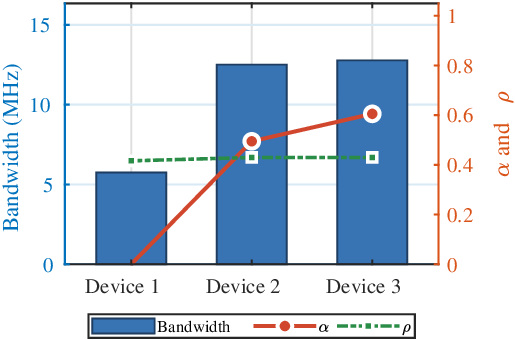}    \label{fig:m8_bandwidth_pruning_insight}}
    \caption{Device importance coefficients, retained-model ratios, and bandwidth allocations for representative video-to-text clusters. Clusters $m=3$ and $m=5$ operate near the no-pruning boundary, whereas clusters $m=4$ and $m=8$ illustrate coupled importance-aware resource allocation in the interior pruning regime.}
    \label{fig:bandwidth_pruning_insight}
    \vspace{10pt}
\end{figure*}

Fig.~\ref{fig:vit_accuracy} shows the classification accuracy under different energy and inference latency thresholds.
The proposed design maintains the best overall accuracy trend across the plotted operating points.
As the energy or inference latency threshold is relaxed, all schemes generally benefit from less aggressive pruning and more feasible communication/computation allocation, but the proposed design makes more effective use of this additional resource slack. This is because the optimization simultaneously adjusts the pruning ratios, GPU scheduling, transmit powers, and bandwidth allocation, rather than fixing one dimension of the system configuration. In contrast, the fixed-power and fixed-bandwidth baselines lose flexibility in communication resource adaptation, while the round-robin and random scheduling baselines cannot fully exploit workload and GPU heterogeneity. The comparison with the equal-importance baseline further shows that modeling heterogeneous view contributions is useful for multi-view image classification, since more informative views can be preserved with higher priority under stringent resource constraints.

Figs.~\ref{fig:video_bleu} and~\ref{fig:video_cider} report the video-to-text captioning performance.
Compared with classification accuracy, the captioning curves are less strictly monotonic because BLEU-4 and CIDEr depend on discrete generated tokens and are more sensitive to sample-level linguistic variations.
Nevertheless, the proposed design generally provides a better captioning quality/resource tradeoff than the fixed power, fixed bandwidth, round-robin scheduling, and random scheduling baselines. The improvement is reflected in both BLEU-4 and CIDEr, suggesting that the optimized resource allocation helps preserve not only local $n$-gram matching but also reference-consensus information in the generated captions.
The gap over the equal-importance baseline is more moderate, which is reasonable because both schemes use the same three-view input structure and the downstream caption generator can partially compensate for small differences in visual feature quality. Even in this case, incorporating heterogeneous importance remains beneficial because it guides pruning decisions toward preserving views with stronger unique or synergistic contributions.

Figs.~\ref{fig:vit_delay_timeline} and~\ref{fig:video_delay_timeline} further visualize the task-level inference latency composition for the server-workload-heterogeneous stress setting.
In these figures, the arrival interval corresponds to the time required for on-device inference and feature uploading, the waiting interval corresponds to GPU queueing before server execution, and the server computing interval corresponds to the execution time on the assigned GPU. For the CIFAR-10 task, the proposed design keeps the maximum completion time below the latency threshold, whereas round-robin scheduling and random scheduling exceed the threshold because heavy server workloads are not balanced across GPU queues. For the video-to-text task, the advantage is more pronounced: the proposed design distributes the server-intensive captioning requests across the two GPUs, while the two heuristic schedules accumulate longer waiting intervals on the bottleneck queue. These timeline results complement the threshold-sweep curves by showing how the optimized scheduling variables affect the internal latency composition of multi-cluster co-inference, especially when the server workloads are heterogeneous.
Overall, the results on both task families show that the proposed framework is not tied to a single evaluation metric, and that joint pruning-aware communication and GPU-scheduling design can improve inference quality under resource constraints.

Fig.~\ref{fig:bandwidth_pruning_insight} compares the optimized importance coefficients, pruning ratios, and bandwidth allocations across four representative video clusters. For clusters $m=3$ and $m=5$, the retained-model ratios $\rho_{m,k}$ are close to one for all devices, and the corresponding bandwidth allocations are nearly uniform within each cluster despite differences in importance coefficients. In contrast, clusters $m=4$ and $m=8$ exhibit lower, interior retained-model ratios and more differentiated bandwidth allocations. In particular, Devices~2 and~3, whose importance coefficients are substantially higher than that of Device~1, receive larger bandwidth shares, whereas Device~1 receives a smaller allocation due to its nearly negligible semantic contribution. These trends reveal that the pruning point determines whether semantic importance can effectively influence bandwidth allocation. When $\rho_{m,k}$ approaches one, as in clusters $m=3$ and $m=5$, model retention is nearly saturated. A higher importance coefficient can no longer induce a meaningful increase in the retained-model ratio, weakening the principal pathway through which semantic importance affects the joint optimization. Bandwidth allocation is therefore governed mainly by communication-related factors, such as upload-delay requirements, channel conditions, transmit power, and cluster-arrival constraints. By contrast, when the pruning ratios lie in the interior region, as in clusters $m=4$ and $m=8$, they remain responsive to semantic importance. The coupling among importance coefficients, model pruning, local computation, upload delay, and the shared bandwidth constraint becomes more active, allowing important devices to receive preferential bandwidth. The key insight is that the importance does not universally translate into additional bandwidth; instead, the pruning point acts as a gate that activates or saturates its influence, while the final allocation results from its joint interaction with communication conditions.

\section{Conclusions}
In this paper, we proposed a pruning-aware multi-cluster LAIM co-inference framework for edge perception networks. The proposed architecture explicitly captures the coupling among communication, computation, and task scheduling arising from multiple device clusters competing for shared wireless spectrum and GPU resources. To characterize the impact of model compression on collaborative inference, we established a rate-distortion-based analytical framework that reveals the trade-off between model pruning and inference performance while incorporating heterogeneous device contributions through PID. Based on this analysis, we formulated a joint inference distortion minimization problem and developed a low-complexity iterative algorithm to jointly optimize model pruning ratios, task scheduling, bandwidth allocation, and transmission power under latency and energy constraints. Extensive simulations on image classification and video-to-text captioning tasks demonstrated that the proposed framework consistently outperforms benchmark schemes, achieving superior inference quality and resource efficiency.
Several directions deserve further investigation. First, the current framework assumes static device clustering and inference requests. Extending the proposed design to dynamic edge environments with user mobility, time-varying network conditions, and changing workloads is an important topic for future research. Second, although this work focuses on pruning-aware co-inference, integrating other model compression techniques, such as quantization and knowledge distillation, may further improve deployment efficiency for resource-constrained edge devices. Finally, adaptive multi-agent collaboration with heterogeneous foundation models is another promising research direction. In such systems, model selection, feature fusion, and resource allocation can be jointly optimized to better coordinate multiple agents with different capabilities, thereby supporting more complex embodied AI and edge perception applications.


\bibliography{AirCompforFL}
\bibliographystyle{IEEEtran}

\end{document}